\documentclass[11.5pt]{article}

\usepackage{amsmath}
\usepackage{natbib}
\usepackage[pdftex]{hyperref}
\usepackage{hypernat}
\usepackage{mathtools}
\usepackage{amsfonts}
\usepackage{dsfont}
\usepackage[T1]{fontenc}
\usepackage{graphicx,psfrag,epsf}
\usepackage{enumerate}

\newcommand{\blind}{0}

\def\spacingset#1{\renewcommand{\baselinestretch}%
{#1}\small\normalsize} \spacingset{1}

\usepackage{tikz}
\usetikzlibrary{arrows,positioning,shapes.geometric}
\tikzstyle{line} = [draw, -latex']

\let\oldbibliography\thebibliography
\renewcommand{\thebibliography}[1]{\oldbibliography{#1}
\setlength{\itemsep}{.8pt}} 

\newcommand{\expect}{\mathsf{E}\expectarg}
\DeclarePairedDelimiterX{\expectarg}[1]{[}{]}{%
  \ifnum\currentgrouptype=16 \else\begingroup\fi
  \activatebar#1
  \ifnum\currentgrouptype=16 \else\endgroup\fi
}
\newcommand{\innermid}{\nonscript\;\delimsize\vert\nonscript\;}
\newcommand{\activatebar}{%
  \begingroup\lccode`\~=`\|
  \lowercase{\endgroup\let~}\innermid 
  \mathcode`|=\string"8000
}

\newtheorem{proposition}{Proposition}
\newtheorem{corollary}{Corollary}
\newtheorem{theorem}{Theorem}[section]
\newtheorem{definition}{Definition}
\newtheorem{lemma}{Lemma}
\newtheorem{proof}{Proof}

\begin{document}

\if 0\blind
{
  \title{\bf A Harris process to model stochastic volatility}
  \author{Michelle Anzarut  and Rams\'es H. Mena}
    \date{}
  \maketitle
  \begin{quote}
	\begin{small}
	\noindent Universidad Nacional Aut\'onoma de M\'exico, Mexico.\\  \textit{E-mail}: michelle@sigma.iimas.unam.mx; ramses@sigma.iimas.unam.mx 
	\end{small}
	\end{quote}
} \fi

\if1\blind
{
  \bigskip
  \bigskip
  \bigskip
  \begin{center}
    {\LARGE\bf A Harris process to model stochastic volatility}
      \date{}
\end{center}
  \medskip
} \fi

\bigskip

\begin{abstract}
	We present a tractable non-independent increment process which provides a high modeling flexibility. The process lies on an extension of the so-called Harris chains to continuous time being stationary and Feller. We exhibit constructions, properties, and inference methods for the process. Afterwards, we use the process to propose a stochastic volatility model with an arbitrary but fixed invariant distribution, which can be tailored to fit different applied scenarios. We study the model performance through simulation while illustrating its use in practice with empirical work. The model proves to be an interesting competitor to a number of short-range stochastic volatility models. 
\end{abstract}

\noindent%
{\it Keywords:}  Harris process, piecewise constant volatility, stationary process, stochastic process estimation, volatility forecasting
\vfill

\newpage
\spacingset{1.44} 

\section{Introduction}\label{section:Intro}

When studying random phenomena evolving in continuous time, Markov processes are the cornerstone models. Much of their application is confined to the subclass of Lévy processes, which are processes with independent, stationary increments. The popularity of the Lévy class is in part due to the existing theory and mechanisms that facilitate their operation. Often, these make possible to perform computations explicitly, and to present challenging results in a simple manner. Nevertheless, empirical data, or theoretical considerations, could suggest the phenomena have a different dependence structure. \citet{bottcher2010feller} documents this issue well, providing examples in areas such as hydrology, geology, and mathematical finance, where state-space dependent models provide a better fit.  

Inducing a state-space dependence in a stochastic process may increase its statistical complexity to a significant degree. The fundamental problem is that, typically, a single realization is observed.  Within the Lévy processes framework, the increments of such a realization form an independent sample of a probability distribution. Thus, statistical inference becomes simpler. Outside the Lévy case, some stability properties can be assumed to make the inference manageable. 

A middle ground between the Lévy processes tractability, and the  Markovian processes generality, arises when one looks at the class of Harris recurrent processes. Harris recurrence means that every non-trivial state is visited infinitely often with probability one \citep[see, e.g.,][]{meyn1993stability}. This class of processes has the advantage of allowing limiting results, which serve as a tool for statistical analysis, while permitting a wide range of sample behaviors. 

In the discrete-time case, a Harris recurrent Markov process is known as a Harris chain. When an aperiodic Harris chain $X=(X_n)_{n \in \mathbb{N}}$ takes values in a separable metric space, the splitting technique \citep{nummelin1978splitting} allows to write its transition functions as a weighted sum between two probability measures, one depending on the starting point, and the other one independent. So for $\epsilon \in (0,1)$,
\begin{align}  \label{transition0}
	 \mathds{P}(X_{n+1}  \in A | X_n = x) = (1-\epsilon)Q(A) + \epsilon \mu_x(A),
\end{align}
where $Q$ is a probability measure, and $\mu_x$ is a probability measure for each $x$. 

Harris chains exhibit a wide-sense regenerative structure. This means there exist certain regeneration times which allow the trajectory of the process to be split into identically distributed cycles. Indeed, representation \eqref{transition0} has the advantage of explicitly marking the regeneration. The process will be dependent on the previous point $x$ with  probability $\epsilon$, or it will take a new value, independent of $x$, with the complimentary probability, thus starting a new cycle. 

The structure of \eqref{transition0}  also clarifies that the chain flexibility relies in two aspects. One is the dependence structure $\mu_x$, and the other is the model distributional properties, build upon $Q$. A trade-off between the generality of distributions $\mu_x$ and $Q$ is usually made since, often, a flexible model can be obtained by focusing solely on the generality of one of them. For instance, assuming a multimodal  $Q$, it is possible to model with stationary processes some trajectories typically associated to non-stationarity \citep[see, e.g.,][]{antoniano2016nonparametric}. Hence, a flexible $Q$ allows a simple choice of $\mu_x$.   In the opposite direction, complex dependence structures may compensate simple choices of distributional features. Such is the case, for example, of several non-linear autorregresive processes.  

Here, we will let $Q$ to be arbitrary, and restrict ourselves to a simple dependence structure, setting $\mu_x$ as a degenerate probability distribution in the value $x$. Denoting the resulting Harris chain by $Y=(Y_n)_{n \in \mathbb{N}}$, we then have
\begin{align} \label{transition_chain} 
	 \mathds{P}(Y_{n+1}  \in A | Y_n = x) = (1-\epsilon)Q(A) + \epsilon \delta_x(A),
\end{align}
where $\delta_x(A)$ is the Dirac measure. When $\epsilon = 0$ transitions \eqref{transition_chain} correspond to an independent process, and $\epsilon = 1$ results in a completely dependent process with constant paths.  
 
We can extend a chain with transition functions \eqref{transition_chain} to a continuous-time Markov process $H = (H_t)_{t \geq 0}$ by making the parameter $\epsilon$ a function of time, $t \rightarrow \epsilon(t)$. Such an extension preserves entirely the chain virtues. The resulting non-independent-increment process will be Harris recurrent, and will model in a simple manner the similarity between observations, while providing a flexible approach.  Performing this extension, the Chapman-Kolmogorov equation leads to the transition probabilities  $P_t(x, A) = \mathds{P}(H_{t}  \in A | H_0 = x)$ given by
\begin{align} \label{transition} 
	P_t(x,A) = (1-e^{-\alpha t})Q(A) + e^{-\alpha t} \delta_x(A), 
\end{align}
for some $\alpha >0$, where $Q$ is a probability distribution. Over the work, we will assume a random starting point $x$ with distribution $Q$.

In continuous time, whether all Harris recurrent Markov processes have a wide-sense regenerative structure is an open problem \citep{glynn2011wide}. This has caused the definition of Harris processes to be somewhat combined in the literature. We will say a continuous-time stochastic process is Harris if it is strongly Markovian, Harris recurrent, and wide-sense regenerative. The extension of Harris chains to \eqref{transition}  falls in this class, since the wide-sense regenerative structure follows automatically. The process has piecewise constant paths, it stays in its current state an exponential time, before jumping to another state randomly sampled from $Q$. Furthermore, in Section \ref{section:Harris_process}, we show that the process with transition probabilities \eqref{transition} is stationary and, in fact, it is the only Feller process generalizing \eqref{transition0} to continuous time. Due to these observations, we will term such a process as the stationary-Feller-Harris process, denoted with \textit{SF-Harris process}. 
 
In this paper, we investigate the SF-Harris process features, and then apply the process to measure the variability level of asset prices in a financial market. The use of piecewise constant volatilities in models has proven to be an efficient modeling technique. For example, \cite{mercurio2004statistical} suppose the volatility is piecewise constant. They address the problem of filtering the intervals of time homogeneity, and then estimate the volatility by local averaging. Also, \cite{davies2009recursive} propose a piecewise constant volatility function, constructing it so that the number of intervals of constancy is minimized.  

The strength of the proposed stochastic volatility (SV) model relies on its mathematical tractability, for (i) it has attractive and flexible stability properties, (ii) it shares appealing characteristics with popular SV models, and (iii) it provides a good approximation to observed market behavior. These while having a simple transition form which allows to understand, estimate, and test the model.   

Continuous-time models offer the natural framework for theoretical option pricing. As a result, they have an important role in the literature since the mid-1980s, being the major process used a Brownian motion timed changed by a subordinator representing the volatility. An overview on SV models is given by \citet{shephard2009stochastic}. Assuming friction-less markets, a weak no-arbitrage condition implies the asset log-price is a semimartingale. This leads to the formulation $Y_t  = A_t+ B_{\tau^*_t}$, where $A = (A_t)_{t \geq 0}$ is a finite variation process, $B = (B_t)_{t \geq 0}$ is Brownian motion, and  $\tau^ *= (\tau^*_t)_{t \geq 0}$ is a time change, all three stochastically independent. A popular choice for $A$ is $A_t = \mu t + \beta \tau^*_t$ for a pair of constants $\mu$ and $\beta$, often referred to as the drift and the risk premium. The time change $\tau^*$ is a positive subordinator, which is a real-valued process with non-negative and non-decreasing sample paths. At the outset, $\tau^*$ was assumed to be a Lévy subordinator \citep{clark1973subordinated}.  However, empirical data suggested that independence of increments is not observed in time-series of returns. Properties, such as volatility clusters, indicate the amplitude of returns is positively autocorrelated in time. Surveys on this matter are given by \citet{bollerslev1994arch, ghysels1996stochastic}; and \citet{shephard1996statistical}. 

To address the non-independence of the volatility increments new continuous-time models have arisen. A popular approach is the diffusion-based models \citep[see, e.g.,][]{hull1987pricing, wiggins1987option}. Such models are in fact special cases of time-changed Brownian motions, where the time change is an integrated process, $\tau^*_t = \int_0^t \tau_s ds$. The process $\tau= (\tau_t)_{t \geq 0}$, identified as the spot volatility, is assumed to have almost surely locally square integrable sample paths, while being positive and stationary. Indeed, integrated processes make sense theoretically, since they are a natural choice to consider the unobserved volatility period which appears when working with a time discretization. 

Considering all these aspects, the resulting equation is 
\begin{align} \label{SVmod}
	Y_t = \mu t + \beta \tau^*_t+ B_{\tau^*_t}, \ \text{where}  \  \tau^*_t = \int_0^t \tau_s ds.
\end{align}
Several models may originate by assuming different processes $\tau$ in \eqref{SVmod}. In particular, the SV model introduced in this paper lies in this class, where $\tau$ is the SF-Harris process with an invariant distribution $Q$ with positive support. Other examples are the standard Black-Scholes model, which is recovered when $\tau$ is a positive constant, and the notorious models proposed in \citet{barndorff2001non} (hereafter termed BNS models), where $\tau$ follows an Ornstein-Uhlenbeck-type  (OU-type) equation. 

Estimating and testing for the different choices of $\tau$ presents a number of difficulties. Novel simulation strategies have been developed over the years which allow us to test some processes $\tau$ in real data \citep[see, e.g.,][]{griffin2010bayesian, griffin2016inference, gander2007stochastic}. Nevertheless, we believe an easily tractable yet flexible model is of significant importance; there is a clear need of a parsimonious model which is good predicting and, perhaps just as importantly, that can be implemented without requiring a lot of user interaction, so that non-experts can run it every day in real applications. All these goals are considered when defining, estimating and testing the model over the work.

Research in the late 1990s has shown more complicated volatility dynamics are needed to model high-frequency return data. Among the most common extensions to model \eqref{SVmod} we can find adding to the volatility a periodic component to model systematic patterns;  adding to the returns a finite activity jump process to model infrequent jumps; and considering estimators robust to microstructure noise. Seminal references for these extensions are \citet{andersen1997intraday,eraker2002theimpact}; and \citet{bandi2006separating} respectively. We will address the first two of such extensions for our SV model, adding a jump and a periodic component. Furthermore, we shall see that, as a byproduct of our proposal, we also generalize a model often applied when considering periodicity. Such a model assumes the volatility, after filtering out the periodicity, approximately constant over each day.  The advantage is that we are thinking on the volatility as a piecewise constant process, with random jumps instead of fixed ones. Still, the estimation method for the constant-over-each-day periodicity is extended in a natural way. 

The rest of the paper continues as follows. We begin in Section \ref{section:Harris_process} by formally presenting the SF-Harris process. We proceed by testing estimation methods in Section \ref{section:estimation}. In Section \ref{section:SV_model} we develop the SV model along with its practical implementation. In Section \ref{section:empirical_analysis} we adjust the SV model to IBM return data. Finally, in Section \ref{section:concluding} we give some concluding remarks and discuss directions for future research. In particular, we give a possible extension of the SV model to a long memory model. The proofs, certain additional results, and intermediate steps of the estimation methods are presented in the Appendix. All the methodology is implemented in the \textbf{R} project for statistical computing.

\section{The SF-Harris process}\label{section:Harris_process}

\subsection{Definition and relevant properties}

This section is devoted to the study of the SF-Harris process, summarizing some of its constructions and important properties. Specifically, we focus on useful properties when modeling spot volatility of asset prices.

\begin{definition}
	The \textnormal{SF-Harris process} $H = (H_t)_{t \geq 0}$ is a stochastic process taking values in a measurable space $(E, \mathcal{E})$, evolving in continuous time, and driven by the transition probability functions \eqref{transition}, where $x \sim Q$.  
\end{definition}

\begin{figure}[p]
    \centering
    \includegraphics[scale = .9]{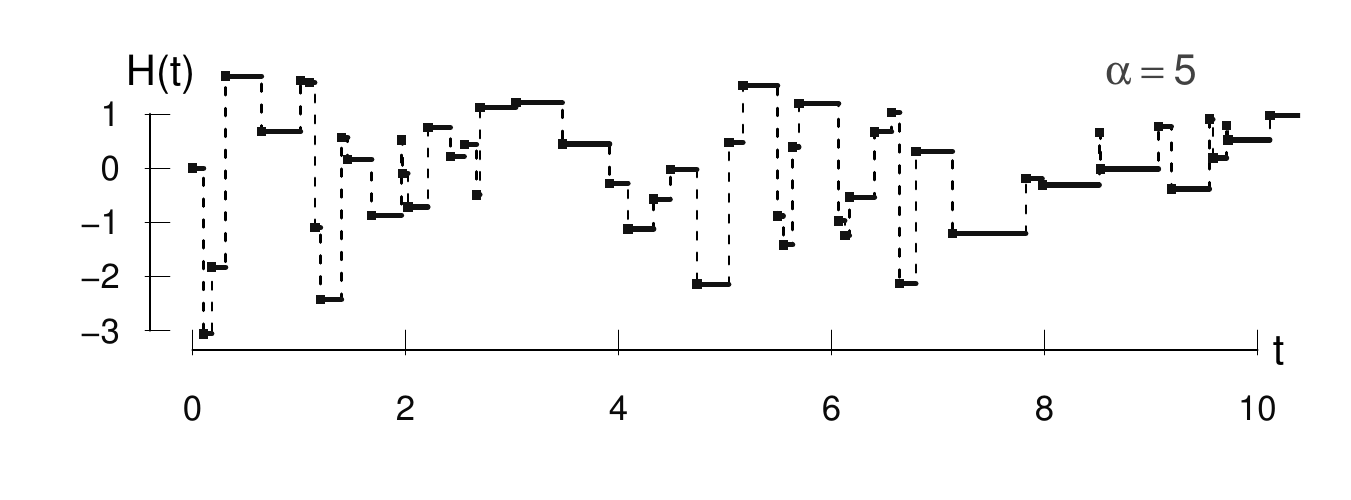}
    \includegraphics[scale = .9]{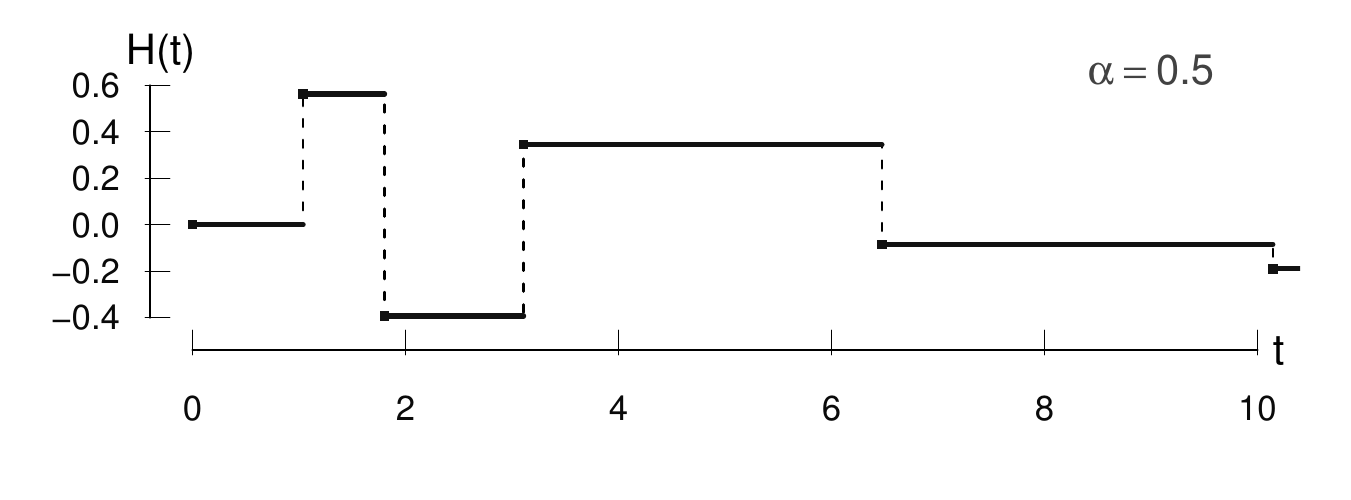}
    \caption{Trajectories of the SF-Harris process $H$, where $Q$ is a standard Normal.}
    \label{fig:alphas} 
\end{figure}

\begin{figure}[p]
    \centering
    \includegraphics[scale = .9]{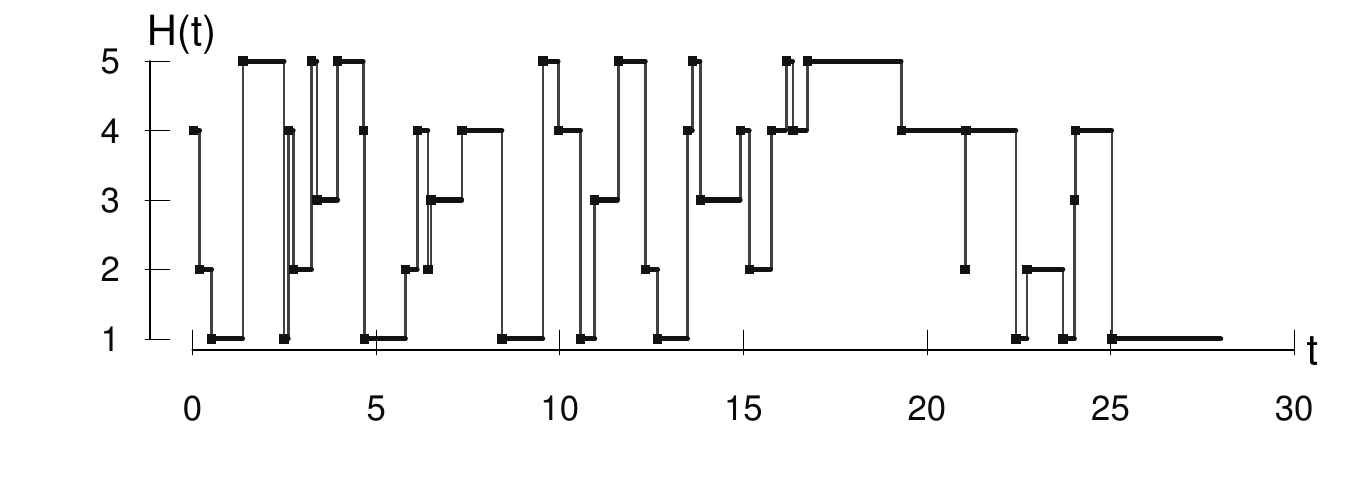}
    \caption{Trajectory of the SF-Harris process $H$, where $Q$ is a $\mathsf{U(1,2,3,4,5)}$ and $\alpha = 2$.}
    \label{fig:unif}
\end{figure} 

Figures \ref{fig:alphas} and \ref{fig:unif} illustrate some SF-Harris process paths; we can observe the distribution $Q$ modulates the marginal behavior at each time, and the parameter $\alpha$ sets  the rate at which the process jumps.  So far, we have stated that the process is Markovian, Harris recurrent, and that it has non-independent increments. The definition of Markov processes is sometimes too general for applications. However, a flexible-enough subclass of Markov processes satisfy the Feller conditions, (i) $\mathsf{T_t}C_0 \subset C_0$ for all $t \geq 0$, (ii)  for any $f \in C_0$, $\mathsf{T_t}f(x) \rightarrow f(x)$ when $t \downarrow 0$, where $C_0$ denotes the space of real-valued continuous functions that vanish at infinity and $\mathsf{T}$ denotes the semigroup operator defined as $\mathsf{T_t}f(x) = \int f(y)P_t(x,dy)$.

Once it is known that a process is Feller, many useful properties follow, such as the existence of right-continuous with left-hand limits paths modifications, the strong Markov property, and right-continuity of the filtration. The SF-Harris process belongs to this class since its semigroup operator is given by
\begin{align} \label{semigroup}
	\mathsf{T_t}f(x) =  (1-e^{-\alpha t}) Qf + e^{-\alpha t} f(x),
\end{align}
where $Qf = \int f(y)Q(dy)$. Hence, it is direct to verify the operator \eqref{semigroup} meets the Feller conditions.

Let us consider what kind of process we would obtain if we extend the wider class of chains \eqref{transition0} to continuous time. The transition probabilities
\begin{align} \label{transition0_process} 
	K_t(x,A) = (1-e^{-\alpha t})Q(A) + e^{-\alpha t} \mu_x(A)
\end{align}
are Markovian, since $K_t(x,A) = \int P_t(y,A)\mu_x(dy)$, where $P_t$ denotes the transition probabilities \eqref{transition} which are Markovian. It also follows the associated processes are wide-sense regenerative. However, only a subclass of \eqref{transition0_process} possesses a unique invariant measure, integrated by the kernels $\mu_x$ which keep $Q$ invariant. Moreover, the semigroup operator corresponding to \eqref{transition0_process} converges to $\int f(y) \mu_x(dy)$ as $t \downarrow 0$. Thus, the semigroup meets the Feller conditions only when $\mu_x = \delta_x$, in which case we return to \eqref{transition}. Due to this, the term SF-Harris process is justified.  

We defined the SF-Harris process through its transition probabilities, but there exists a unique Feller process with right-continuous with left-hand limits paths driven by transition probabilities  \eqref{transition}. An explicit representation can be found by uniformizing the chain $Y$ with transition probabilities \eqref{transition_chain}. Let us expand on this. Assuming $\epsilon \in (0,1)$, consider the process $Y_N = (Y_{N_t})_{t \geq 0}$, where $N = (N_t)_{t \geq 0}$ is a Poisson process of rate $\lambda > 0$ stochastically independent of $Y$. 

It is easy to prove inductively that the k-step  transitions of $Y$ are given by
\begin{align*}
	 \mathds{P}(Y_{k}  \in A | Y_0 = x) = (1-\epsilon^k)Q(A) + \epsilon^k\delta_x(A).
\end{align*}
Therefore,
\begin{align*}
	\mathds{P}(Y_{N_t}  \in A | Y_0 = x) &= \sum_{k=1}^\infty  \mathds{P}(Y_{k}  \in A | Y_0 = x)\dfrac{(\lambda t)^k e^{-\lambda t}}{k!} \\
	&= \left(1-e^{-\lambda t (1-\epsilon)}\right)Q(A) + e^{-\lambda t (1-\epsilon)} \delta_x(A).
\end{align*}
By making $\lambda = \alpha(1-\epsilon)^{-1}$, we recover \eqref{transition}. Consequently, we have the following stochastic representation which is useful for simulating and estimating the process.

\begin{theorem}  \label{prop:explicit_form}
	Let $Y = (Y_n)_{n \in \mathbb{N}}$ be a Markov chain with transition functions given by \eqref{transition_chain}, where $\epsilon \in [0,1)$, $Q$ is a probability distribution, and $x \sim Q$. If  $N= (N_t)_{t \geq 0}$ is a Poisson process with rate $\lambda=\alpha(1-\epsilon)^{-1}$ independent of $Y$, then the SF-Harris process, with marginal distribution $Q$ and jump parameter $\alpha$, has the stochastic representation $H_t = Y_{N_t}$ for all $t \geq 0$. 
\end{theorem}

Theorem \ref{prop:explicit_form} implies the SF-Harris process is a pseudo-Poisson process, defined by \cite{feller1966introduction} as a continuous-time process that can be obtained from Markov chains by subordination with a Poisson process. The case $\epsilon = 1$ corresponds to the limit $\alpha \rightarrow \infty$, where no jumps occur so the process has constant paths. The case $\epsilon = 0$ corresponds to independent identically distributed random variables $(Y_n)_{n \in \mathbb{N}}$, it is included in the theorem taking into account that  $P^0(x,A)=\delta_x(A)$.

Notice that the parameter $\lambda = \alpha(1-\epsilon)^{-1}$ ranges in the positive real line. Fixing a value for $\alpha$, when $\epsilon$ grows $\lambda$ gets smaller. In other words, when the dependence in the paths of the Harris chain $Y$ grows, then the Poisson process $N$ introduces less dependence, in order to get exactly the same dependence rate $\alpha$ in the uniformized process $H = Y_N$. For the construction or simulation of a SF-Harris process using Theorem \ref{prop:explicit_form} any value $\epsilon$ can be used. In particular, we stick to the simpler case $\epsilon = 0$ which results in $\lambda = \alpha$. The following corollary is useful for some calculations.

\begin{corollary}  \label{corol:explicit_form}
	Let $Y$ be as in Theorem \ref{prop:explicit_form}, and let $(T_n)_{n \in \mathbb{N}}$ be a sequence of random variables independent of $Y$ whose increments $(S_n)_{n \in \mathbb{N}}$, defined as $S_n = T_n-T_{n-1}$, are exponential, independent, and identically distributed variables with mean $(1-\epsilon)/\alpha$, i.e., $S_n \sim \mathsf{Exp}\left\{ \alpha(1-\epsilon)^{-1}\right\}$. Then the SF-Harris process, with marginal distribution $Q$ and jump parameter $\alpha$, has the stochastic representation
	\begin{align} \label{eq:defX}
		H_t = x \mathds{1}_{t<T_1}+\sum\limits_{n=1}^\infty Y_n \mathds{1}_{t \in {[T_n, T_{n+1})} }, \ \text{for all} \ t \geq 0.
	\end{align}
\end{corollary}

As of yet, we have three representations that give an intuition about the SF-Harris process, and facilitate the study of distinct features. For instance, using the transition probability functions \eqref{transition}, it can be easily checked the SF-Harris process is time-reversible with invariant measure $Q$, and so a strictly stationary process. On the other hand, with the representation in Corollary \ref{corol:explicit_form} it can be shown that for any $B$ such that $Q(B)>0$, we have that $\int_0 ^\infty \delta_{H_t}(B)dt $ is infinite almost surely. This means that the state $B$ is visited infinitely often with probability one, since $Q$ is a finite measure, the process is positive Harris recurrent. 

As a consequence of the Harris recurrence, the process will eventually converge to $Q$. An important measure of the speed of convergence is the total variation distance  given in this case by
\begin{align*}
	\sup_{A \in \mathcal{E}} |P_t(x,A)-Q(A)| =e^{-\alpha t} \sup_{A \in \mathcal{E}}  |\delta_x(A)-Q(A)| =e^{-\alpha t},
\end{align*}
which translates to the process being uniformly ergodic.   Using the semigroup operator \eqref{semigroup} it is straightforward to check the SF-Harris process has constant mean and variance which match the ones of the distribution $Q$. Lastly, the process regenerative property is evidenced in a natural way from Corollary \ref{corol:explicit_form} taking $\epsilon = 0$. The independent and identically distributed $(S_n)_{n \in \mathbb{N}}$ are inter-regeneration times were the process starts afresh. Hence, they form a renewal process, and split the stochastic process into a sequence of identically distributed cycles. 
 
These stability properties represent a crucial component in the theory and application of stochastic processes. Indeed, in several modeling contexts, the assumption that some distributional features remain invariant over time is often needed to implement estimation and prediction procedures, or simply to be able to analytically determine quantities of interest. Furthermore, in most cases, no general expression for the transition probabilities is available, particularly in the continuous-time setting. The full control of the transition probabilities driving the SF-Harris process, added to the stability in the process behavior, allow us to understand the process dynamics, and to develop methods for the process application in different contexts, particularly, in the context of stochastic volatility.

\subsection{Integrated SF-Harris process} \label{section:int_process}

In the proposed SV model of this paper, which we develop further in Section \ref{section:SV_model}, the SF-Harris process is used as a spot volatility process. Therefore, the log-prices volatility will match the integrated SF-Harris process $H^* = (H^*_t)_{t\geq 0}$, where $H_t^* = \int_0^{t} H_s ds$. The features shown in the previous section about $H$, allow us to derive in a simple manner similar properties for $H^*$. For example, the stochastic representation in Corollary \ref{corol:explicit_form} provides the means to derive a similar representation for $H^*$. 
\begin{theorem} \label{repH_int} 
	Let $Y$, $(T_n)_{n \in \mathbb{N}}$, and $(S_n)_{n \in \mathbb{N}}$ be as in  Corollary \ref{corol:explicit_form} . If $N_t = \max\{n \in \mathbb{N}:T_n \leq t \}$, then $H^*$ has the stochastic representation
\begin{equation*}
	H^*_t = \sum\limits_{n=0}^{N_t-1} (Y_n-Y_{N_t}) S_{n+1} + Y_{N_t}t,
\end{equation*}
where we let $Y_0$ and $T_0$ be constant random variables equal to $x$ and $0$ respectively.
\end{theorem}
Note that, once again, $N = (N_t)_{t\geq 0}$ is a Poisson process of rate $\lambda=\alpha(1-\epsilon)^{-1}$. 

Now, if  $\xi = \expect*{H_t}$, it is straightforward to see that $\expect*{H^*_t}=\xi  t$. Moreover, the ergodicity of $H$ implies that, as $t \rightarrow \infty$, $t^{-1}H^*_t \rightarrow \xi $ almost surely and, consequently, the returns tend to normality. This desirable result is known as aggregational Gaussianity \citep{barndorff2003integrated}. Also, second order moments and correlations for both the integrated process and the returns can be derived explicitly following \citet{barndorff2001non}, where they deal with the general case of the process $\tau$ in \eqref{SVmod} being any second order stationary process.

 \subsection{A semi-Markovian extension} \label{section:semi}
 
When applying the SF-Harris process to model volatility, the resulting SV model is a short-memory model since, for any positive times $t, h$, $\mathsf{Cor}(H_t,H_{t+h}) = e^{-\alpha h}$. To capture more realistic dependence structures we could consider to model the time between jumps with a heavy-tailed distribution, rather than an exponential one. That is, taking $Y$as in Theorem \ref{prop:explicit_form}, we can define a more general process
$\xi = (\xi_t)_{t \geq 0}$, given by
	\begin{align} \label{eq:semi-markov}
		\xi_t = x I_{t<R_1}+\sum\limits_{n=1}^\infty Y_n I_{t \in {[R_n, R_{n+1})} },
	\end{align}
where $(R_n)_{n \in \mathbb{N}}$ is a sequence of random variables independent of $Y$ whose increments $(V_n)_{n \in \mathbb{N}}$, defined as $V_n = R_n-R_{n-1}$, are positive with probability one, independent, identically distributed, and with an arbitrary cumulative distribution function $G$. 

When the increments $(V_n)_{n \in \mathbb{N}}$ are exponential, we recover the SF-Harris process, in which case the lack of memory property of the exponential distribution causes the process to be Markovian. Otherwise, the process will be non-Markovian. However,
\begin{align*}
	P(Y_{n+1} \in A, R_{n+1}  \leq t &| Y_m = y_m, R_m = t_m, m \le n)\\
	& = P(Y_{n+1} \in A, R_{n+1} - R_n \leq t-t_n | Y_n = y_n).
\end{align*}
Hence, the process can be embedded in a Markov process on a higher dimensional state space.  This implies the process  \eqref{eq:semi-markov} is a particular case of a semi-Markov processes \citep{levy1954processus}. 

A characteristic feature of the process  \eqref{eq:semi-markov} and, in fact, of any semi-Markov process, is a set of intervals of constancy in their trajectory. Since this structure is similar to the SF-Harris process we can deduct a number of properties. For example,  the process  \eqref{eq:semi-markov} has non-independent increments, its mean and variance match the ones of the distribution $Q$, while its autocorrelation function is given by $r(t) = 1-G(t)$. We also have that, for any $A \in \mathcal{E}$,
 \begin{align*}
 	\mathds{P}_x(\xi_t \in A) = G(t) Q(A) +  \{1-G(t)\} \delta_x(A).
 \end{align*}
  Therefore, the process is wide-sense regenerative. It is immediate to see as well that $\mathds{P}_x(\xi_t \in A) \rightarrow Q(A)$ when $t \rightarrow \infty$. This implies $Q$ is the limit distribution. Other stability properties can be deduced easily, such as ergodicity, positive Harris recurrence, and representations  analogous to the ones we obtained for the SF-Harris process. The details appear in Appendix \ref{App:stability}. Such stability properties make the tractability of the model possible, and become important when working on the inference and application to real data sets. Later, we will use special cases of the process  \eqref{eq:semi-markov} to extend the memory of the SV model proposed in this paper.

\section{Estimation and prediction}\label{section:estimation}

\subsection{Estimation}\label{subsection:est-methods}

In this section we develop an efficient estimation method for the SF-Harris process based on a Gibbs sampler. Additionally, we provide other three methods for its comparison. We consider the one-dimensional case, although all methods can be extended in a natural way to distributions $Q$ in higher dimensions. 

We assume that $x_1,...,x_n$ is a discrete realization of the process at times $\hat{t}_0< \dots <\hat{t}_n$, we denote with $t_i = \hat{t}_i - \hat{t}_{i-1}$ for $i=1,...,n$, $t_0 = 0$, we let $\beta$ be the set of parameters of $Q$, and assign independent prior distributions $\pi(\alpha), \pi(\beta)$. Since $\alpha$ is a positive real number, we choose $\pi(\alpha) = \mathsf{Exp}(c) $ for some $c>0$, the choice of $\pi(\beta)$ will depend on $Q$. Consequently, the joint posterior distribution is
\begin{align*}
	\pi(\alpha, \beta | x_1,...,x_n) \propto \prod_{i=1}^n \left\{  (1-e^{-\alpha t_i}) Q(x_i|\beta)+e^{-\alpha t_i} \delta_{x_i}(x_{i-1}) \right\}e^{-\alpha c}\pi(\beta).
\end{align*}	
This distribution has $2^n$ terms. Rather than dealing with it, we work with the latent random variables $z_0,...,z_n$, where $z_i$ equals one if $x_i$ comes from $Q$, or zero otherwise. Denoting with $\textbf{x} = (x_1,...,x_n)$ and $\textbf{z} = (z_1,...,z_n)$, we have that 
\begin{align}\label{eq:post_beta}
\pi(\beta | \alpha, \textbf{x}, \textbf{z}) &\propto  \prod_{i=1}^n \left\{ (1-e^{-\alpha t_i}) Q(x_i|\beta)\right\}^{z_i}  \pi(\beta),
\end{align}
so the choice of $ \pi(\beta)$ can cause \eqref{eq:post_beta}  to have a tractable form. To simulate from the full conditional distribution of $\alpha$ we consider two options:
\begin{itemize} 
	\item[(a)] (Gibbs-a) We apply Adaptive Rejection Metropolis Sampling \citep{gilks1995adaptive} to simulate directly from the full conditional distribution of $\alpha$.
	\item[(b)] (Gibbs-b) We use an additional set of latent random variables. Let  $m = \sum_{i=0}^n z_i$, and $j_0,...,j_m$ be the ordered times where the variables $z_i$ equal one. Then $j_1-j_0,j_2-j_1,...,j_m-j_{m-1}$ are independent, identically distributed variables with distribution $\mathsf{Exp}(\alpha)$, so it turns out  $\alpha |(j_0,...,j_m) \sim \mathsf{Gamma}(m+1, j_m+c)$.
\end{itemize}
The two options of the Gibbs sampler, Gibbs-a and Gibbs-b, are compared through a simulation study to other three inference methods, the no difference no jump method (NDNJ), maximum likelihood estimation (MLE), and the expectation-maximization algorithm (EM). When performing the comparison, we use the posterior distributions modes as estimators. 

In the NDNJ method, we assume that when two consecutive observations are different, it occurs because the process jumped exactly at that moment, and when they are the same, it occurs because the process did not jump. Following this reasoning, we equate the expected value of $\alpha$ with its sample mean, and then estimate the parameters of $Q$ with the observations that are distinct. This method is likely to show estimation problems when $Q$ has a high probability of repeated values, or when few observations are available. 

The MLE consists in the numerical maximization of the  log-likelihood function.
\begin{align*}
	\log L(\alpha, \beta) &= \log Q(x_0|\beta) + \sum_{i=1}^n \log \left\{ (1-e^{-\alpha t_i}) Q(x_i|\beta) + e^{-\alpha t_i}  \delta_{x_i}(x_{i-1}) \right\},
\end{align*}

Finally, for the EM, we use again the latent random variables $z_0,...,z_n$. As a consequence, the augmented log-likelihood results in
\begin{align*}
	\log L^a(\alpha, \beta) &=  \sum \limits_{i=1}^n \left\{ z_i\log(1-e^{-\alpha t_i}) +(1-z_i)(-\alpha t_i) \right\} + \sum \limits_{i=0}^n z_i\log Q(x_i|\beta).
\end{align*}
Now $z_i|(x_i, \alpha, \beta) \sim \mathsf{Ber}(p_i)$ with
\begin{align*}
	p_i = \frac{(1-e^{-\alpha t_i}) Q(x_i|\beta)}{(1-e^{-\alpha t_i}) Q(x_i|\beta)+e^{-\alpha t_i} \delta_{x_i}(x_{i-1})}.
\end{align*}
Denoting by $p_i^{(k-1)}= \mathsf{E}_{z_i|x_i, \alpha^{(k-1)}, \beta^{(k-1)}}[z]$, at each step of the EM we maximize the function $ G\left\{\left(\alpha, \beta\right)| \left(\alpha^{(k-1)}, \beta^{(k-1)}\right)\right\}$, given by
\begin{align*}
	\sum \limits_{i=1}^n \left\{ p_i^{(k-1)}\log(1-e^{-\alpha t_i}) +(1-p_i^{(k-1)})(-\alpha t_i) \right\}  
	 + \sum \limits_{i=0}^n p_i^{(k-1)}\log Q(x_i|\beta),
\end{align*}
either numerically or analytically; depending on the form of $Q$.

\subsection{Simulation study}\label{subsection:est-simulation_study}

We present a simulation study to account for the performance of the estimation method. The main characteristics affecting the performance are the number of available observations, and the probability of repeated values in $Q$. To exemplify this, we vary the number of observations on the samples, and test all the methods with two marginal distributions. The first one is a discrete uniform on the values $1,2,3,4,5$, denoted by $\mathsf{U(1,2,3,4,5)}$. The second one is a generalized inverse Gaussian (GIG) \citep[see, e.g.,][]{eberlein2004generalized}, with density given by
\begin{align*}
		 \mathsf{GIG}(\lambda, \kappa, \eta) \propto x^{\lambda-1} \exp \left\{ -\frac{\kappa}{2}\left(\frac{\eta}{x} + \frac{x}{\eta}\right)\right\}, \ \text{for} \ x>0,
\end{align*}
where $\lambda \in \mathbb{R}$, $\kappa \geq 0$, $\eta > 0$, and $\kappa > 0$ when $\lambda = 0$.

Roughly, the testing  procedure comprises two steps:

\begin{itemize}
	\item[I.] We choose 100 parameter values randomly. This is done by simulating from a uniform distribution at reasonable intervals for the parameters.
	\item[II.] For the sample sizes $k$ = 20, 100, 500, and 1000:
	\begin{itemize}
	\item[i.] We simulate a process sample of length $k$ for each parameter value.
	\item[ii.] We estimate each parameter value with the four methods.
	\item[iii.] We calculate an estimation error for each method using the accuracy for the 100 parameters estimation.	
	\end{itemize}			
\end{itemize}

A testing procedure of this kind has a number of advantages. First, it proves the estimation is working for a wide range of parameter sets and for different combinations of them. Second, it assures that any user can run the method without problems since the starting conditions are computed at random. Third, it provides a highly accurate estimation error  because there is no subjectivity involved neither in which set of parameters is selected to display, nor in which starting point to use for each run.

\begin{table}
	\caption{\label{unif_errors}Estimation errors for the SF-Harris process  parameters in the case of $Q$ being $\mathsf{U(1,2,3,4,5)}$}
	\begin{center}
	\begin{tabular}{ c | c c c c c } 
		Sample size & $E_\alpha^{NDNJ}$& $E_\alpha^{MLE}$& $E_\alpha^{EM}$ & $E_\alpha^{Gibbs-a}$ & $E_\alpha^{Gibbs-b}$\\ [.4ex] 
		\hline\\[-3.5ex] 
		20 & 0.45 & 1.13 & 0.83 & 0.31 & 0.44\\
		100 & 0.40 & 0.11 & 0.12 & 0.13 & 0.39\\
		500 & 0.27 & 0.04 & 0.04 & 0.04 & 0.26\\
		1000 & 0.22 & 0.03 & 0.03 & 0.02 & 0.18\\
		\hline
\end{tabular}
\end{center}
\end{table}

We start the estimation with the $\mathsf{U(1,2,3,4,5)}$ case. Figure \ref{fig:unif} illustrates a 28-day trajectory. It is natural to expect that the NDNJ method and the Gibbs-b underestimate the value of $\alpha$. This is because the process is likely to jump and fall back to the previous state. However, in the NDNJ or Gibbs-b methods, if two subsequent observations are equal, we assumed that it occurs because the process did not jump. In step I, we draw the parameters $\alpha_1,...,\alpha_{100}$  from a uniform distribution over $(0,30)$. In step II, we compute the error as $E_\alpha = \frac{1}{100} \sum \limits_{i=1}^{100}\left| \frac{\alpha_i-\hat{\alpha}_i}{30} \right|$. The results are shown in Table \ref{unif_errors}. Gibbs-a outperforms all methods. 

\begin{figure}[p]
    \centering
	\includegraphics[scale = .9]{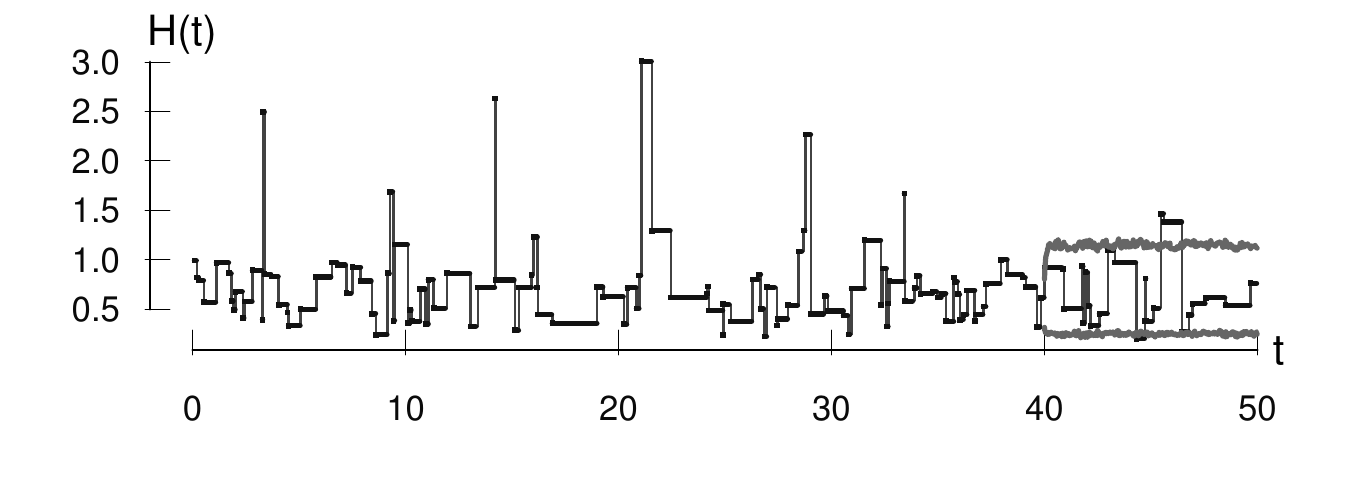}
    \caption{Trajectory of the SF-Harris process $(H(t))_{t \geq 0}$, where Q is a $\mathsf{GIG}(-2, 4, 1)$ and $\alpha = 3$; plus $0.9$-HPD intervals computed with the first 40 days.}
    \label{fig:GIG_pred}
\end{figure} 

We continue with the case of $Q$ being a $\mathsf{GIG}(\lambda, \kappa, \eta)$ distribution. To illustrate the process dynamics we simulate a 40-day trajectory, appearing in Figure \ref{fig:GIG_pred}. The estimation with the four methods may be developed further using the properties of the GIG family, this is located in Appendix \ref{App:estGIG}. We follow the testing procedure again. In step I, we simulate the values $\alpha_1,...,\alpha_{100}$, $\lambda_1,...,\lambda_{100}$, $\kappa_1,...,\kappa_{100}$, and $\eta_1,...,\eta_{100}$ from a uniform distribution over the intervals $(0,30)$, $(-5,5)$, $(0,50)$, and $(0,4)$ respectively. In step II, the errors are calculated as:
\begin{align*}
	E_\alpha = \frac{1}{100} \sum \limits_{i=1}^{100}\left| \frac{\alpha_i-\hat{\alpha}_i}{30} \right| \ \text{and} \	E_Q = \frac{1}{100} \sum \limits_{i=1}^{100} \mathsf{KL}\left\{\mathsf{GIG}(\lambda_i, \kappa_i, \eta_i), \mathsf{GIG}( \hat{\lambda}_i, \hat{\kappa}_i, \hat{\eta}_i)\right\},
\end{align*}
where $\mathsf{KL}$ denotes the Kullback Leibler divergence, computed  in Appendix \ref{App:KL}. The results are presented in Table \ref{errors_gig}. In this case, the NDNJ method provides the best results. The MLE, EM, and Gibbs-a methods are poor for the estimation of $\alpha$, while the Gibbs-b performs well for large samples. 	

\begin{table}
	\caption{ \label{errors_gig} Estimation errors for the SF-Harris process parameters in the case of $Q$ being $\mathsf{GIG}(\lambda, \kappa, \eta)$}
	\begin{center}
	\begin{tabular}{c | c c c c c }
		Sample size & $E_\alpha^{NDNJ}$& $E_\alpha^{MLE}$& $E_\alpha^{EM}$ & $E_\alpha^{Gibbs-a}$ & $E_\alpha^{Gibbs-b}$\\ [0.4ex] 
		\hline\\[-3.5ex] 
		20 & 0.45 & 1.72 & 1.64 & 0.28 & 0.45\\
		100 & 0.39 & 1.54 & 1.14 & 0.16 & 0.38\\
		500 & 0.23 & 1.55 & 1.13 & 0.91 & 0.21\\
		1000 & 0.16 & 1.51 & 1.11 & 1.41 & 0.14\\		[0.6ex] 
		\hline\\[-2.2ex]
		 & $E_Q^{NDNJ}$& $E_Q^{MLE}$& $E_Q^{EM}$ & $E_Q^{Gibbs-a}$ & $E_Q^{Gibbs-b}$\\ [0.4ex] 
		\hline\\[-3.5ex] 
		20 & 0.09 & 0.06 & 0.09 & 0.69 & 0.83\\
		100 & 0.02 & 0.02 & 0.02 & 0.28 & 0.30\\
		500 & 0.01 & 0.02 & 0.03 & 0.17 & 0.19\\
		1000 & 0.01 & 0.03 & 0.03 & 0.13 & 0.14 \\
		\hline
\end{tabular}
\end{center}
\end{table} 

Notice that an immediate procedure for obtaining a set of $m$ trajectories in future times is available by, first,  simulating $m$ parameter values from the posterior distributions, and, second, for each one of the $m$ values, simulating a realization of $H$ starting on $x_n$ at time $t_n$. Given a SF-Harris process trajectory in the case where $Q$ is $\mathsf{U(1,2,3,4,5)}$, we draw a large number of realizations in future times employing the Gibbs-a method, obtaining a set from which we can make predictions. We could compute highest posterior density intervals (HPD intervals), but that would not be truly meaningful since, after the process jumps, it will randomly fall in some state $1,2,3,4,5$. Hence, the most interesting aspect is when the next jump will occur, or the probability of the first jump being before certain time. Indeed, the mean value of the first jump time in the Figure \ref{fig:unif} trajectory is $28.60$, and the probability of the first jump being before one day is $0.79$. In the GIG density case, it does make sense to apply the prediction procedure to find HPD intervals. As an example, we took away the last ten days from the trajectory in Figure \ref{fig:GIG_pred} and performed this; simulating 1000 realizations using the Gibbs-b method and calculating with them HPD intervals. 

\section{Stochastic volatility model}\label{section:SV_model}

\subsection{Definition and relevant properties}

In this section we derive important properties of the proposed SV model, in which log-prices follow \eqref{SVmod}, being $\tau$  the SF-Harris process and $Q$ a distribution with positive support.  Such a model shares several attractive properties with the BNS model. In the BNS model, the spot volatility follows a stationary process of the OU-type; with the restriction that the background driving Lévy process (BDLP) has positive increments and no drift. This implies that the spot volatility jumps when the BDLP does, and decays exponentially in-between. Therefore, in both models the spot volatility is a stationary and ergodic process, with positive jumps, and with an exponential autocorrelation function. 

Moreover, in the BNS model, given a self-decomposable distribution, there is a unique BDLP that will generate that specific marginal distribution for the volatility. Hence, when the marginal distribution Q in the SF-Harris process is self-decomposable, we can choose a BDLP in a way that the OU-type and the SF-Harris process share also the mean and variance. As a consequence, the integrated volatility and the return processes will also be equivalent up to second order. Many of the powerful results developed for estimating and forecasting returns rely only on second moments so they can be equally applied to the SF-Harris process. Examples of these results may be found in \citet{barndorff2001non} or \citet{sorensen2000prediction}.

To implement the proposed SV model, we will restrict our attention to the subclass of models in which the spot volatility has a marginal GIG distribution.

\begin{definition} 
	We term the \textnormal{GIG-Harris SV model} to \eqref{SVmod}, where $\mu$ and $\beta$ are constants, $B=(B_t)_{t \geq 0}$ is a Brownian motion independent of $\tau =(\tau_t)_{t \geq 0} $, and $\tau$ is the SF-Harris process in the case of $Q$ being a GIG distribution.
\end{definition}

Although any positive marginal distribution $Q$ defines a spot volatility model, we choose GIG distributions for various reasons. First, special cases have been extensively used, such as the Gamma, Inverse Gamma, or Inverse Gaussian. Second, they have been proven to accurately fit real data \citep[see, e.g.,][]{gander2007stochastic}. Third, they can adjust to different scenarios since the choice of parameters will change their shape, skewness, and tail weight.

\subsection{Estimation and prediction}\label{subsection:SV_model-estimation}

In this section we develop methods to sample from the posterior distributions of the parameters of the GIG-Harris SV model; based in the estimation method we developed and tested for the SF-Harris process in Section \ref{section:estimation}. Suppose we observe a high-frequency, discretized realization of the log-price process $Y$ at times $t_1<\dots<t_n$. Then, the model is adjusted to returns, defined as $ R_i = Y_{t_i}-Y_{t_{i-1}}$, for $i = 1,...,n$, where we assume $Y_{t_0} = 0$. Given such returns, a series of filters are required to obtain a set of observations of the SF-Harris process. These are summarized in the following diagram.

\vspace{3mm}
\begin{center}
\begin{small}
\begin{tikzpicture}[>=latex']
    
          \tikzset{block/.style= {draw, rectangle, align=center,minimum width=2cm,minimum height=1cm},
        rblock/.style={draw, shape=rectangle,rounded corners=.8em,align=center,minimum width=3.1cm,minimum height=2cm},
        }

        \node [rblock]  (uno) {Given the returns\\ $(R_{t_i})_{i = 1}^n$};
        \node [block, right =.5cm of uno] (dos) {We filter to measure the \\ process of  integrated \\ volatility  $(H^*_{t_i})_{i = 1}^n$};
        \node [block, right =.5cm of dos] (tres) {We filter to measure the \\ process of spot \\ volatility  $(H_{t_i})_{i = 1}^n$};
        \node [block, below =.5cm of tres] (cuatro) {We estimate $\alpha$ and  the \\ parameters of Q: $\lambda,\kappa, \eta$};
        \node [block, below =.5cm of dos] (cinco) {We estimate \\ $\mu$ y $\beta$};
 
        \path[draw,->] (uno) edge (dos)
                    (dos) edge (tres)
                    (tres) edge (cuatro)
         			 (dos) edge (cinco)
                    ;
\end{tikzpicture}
\end{small}
\end{center}
\vspace{.5mm}

The measurement of the integrated volatility process is performed using a common procedure based on the semimartingales quadratic variation. The procedure, found in \citet{barndorff2004measuring}, consists in approximating the quadratic variation of $Y$ with the realized variance. Indeed, the quadratic variation of the GIG-Harris SV model matches the Brownian motion subordinator $H^*$.  This method has proven to work well for irregular time intervals, and when the subordinator is continuous, as is the integrated SF-Harris process case. Now, given the integrated volatility observations, we perform the measurement of the spot volatility process by applying the common right-hand side derivative approximation.
  
In the literature, there are multiple alternatives to filter both the integrated and spot volatility. In order to choose the best method to apply, we tried out a number of these alternatives, comparing them by computational time, and by the estimation errors obtained, which are computed with 100 simulated trajectories and add up the errors of all of the estimation steps (see Section \ref{section:simulation_study}).  The errors show that the path properties, such as the jump rate and the values it ranges in, are well preserved with these simple filtering methods. If, however, we are interested in using the filtration for other purpose than the estimation of the SF-Harris process parameters, other methods could be tried out. For instance, jump robust methods for filtering univariate integrated and spot volatility are the bipower realized variation, which is used later on for the jump detection procedure, and the nonparametric kernel methods of \cite{bandi2009nonparametric}. 

Having the spot volatility observations, we proceed using such an approximation to estimate $\alpha$ and the parameters of $Q$, by means of the Gibbs-b method of Section \ref{subsection:est-methods}. Finally, for the estimation of $\mu$ and $\beta$, we have that  
\begin{align*}
	 R_i \sim N\{\mu(t_i-t_{i-1}) + \beta H_i^*, H_i^*\},
\end{align*} 
where $H_i^* = H^*_{t_i}-H^*_{t_{i-1}}$. Therefore, assigning a Gaussian prior distribution for $(\mu, \beta)$, we obtain a Gaussian posterior distribution with certain updated parameters, detailed in Appendix \ref{App:estimation_mu_beta}.

\subsection{Simulation study} \label{section:simulation_study}

To validate the estimation procedure we perform a simulation study. The testing procedure we employ is similar to the one presented in Section \ref{subsection:est-simulation_study}. We begin by drawing 100 values for $\mu$, $\beta$, $\alpha$, $\lambda$, $\kappa$, and $\eta$ from a uniform distribution over the intervals $(-2,2)$, $(-2,2)$, $(0,30)$, $(-5,5)$, $(0,50)$, and $(0,4)$ respectively. With these parameter values, we proceed with the simulation of 100 series of spot volatility, integrated volatility, and returns. Subsequently, we run all of the estimation steps, and calculate the estimation errors as: 
\begin{align*}
	E_\mu &= \frac{1}{100} \sum \limits_{i=1}^{100}\left| \frac{\mu_i-\hat{\mu}_i}{4} \right|, \ 	
	E_\beta = \frac{1}{100} \sum \limits_{i=1}^{100}\left| \frac{\beta_i-\hat{\beta}_i}{4} \right|, \ 	
	E_\alpha = \frac{1}{100} \sum \limits_{i=1}^{100}\left| \frac{\alpha_i-\hat{\alpha}_i}{30} \right|, \ 	\text{and} \\
	E_Q &= \frac{1}{100} \sum \limits_{i=1}^{100} \mathsf{KL}\left\{\mathsf{GIG}(\lambda_i, \kappa_i, \eta_i), \mathsf{GIG}( \hat{\lambda}_i, \hat{\kappa}_i, \hat{\eta}_i)\right\}.
\end{align*}

The resulting errors are $E_\mu = 0.18$, $E_\beta = 0.10$, $E_\alpha = 0.22$, and $E_Q = 0.20$. The first three are easily interpretable since they are weighted by each interval length, for the interpretation of the error of $Q$ see Appendix \ref{App:KL}. The approximation proves to be quite good, even tough we are using again a rigorous estimation test.

\section{Empirical Analysis}\label{section:empirical_analysis}

\subsection{Data}

In this section, we apply the GIG-Harris SV model to the stock prices of IBM (International Business Machines). Let us start with an explanation of the data.  The three-year series we use, obtained from \cite{Kibot}, covers the period from January 2012 until December 2014. It records at every minute the open, high, low, and close prices, and the volume of IBM stocks; in the regular times of the US trading session, between 9:30 AM and 4:00 PM on workdays. The data are provided adjusted, using appropriate split and dividend multipliers adhering to the Center for Research in Security Prices standards. The procedure to clean the data is detailed  step by step in Appendix \ref{App:cleaning_procedure}.

Performing an initial exploratory analysis we can notice that, as is common in practice, the estimation accuracy should be better when introducing a jump and a periodic component to the model. For this reason, we continue as follows. First, we introduce a jump component to the returns and approximate the integrated and spot volatility processes. Next, we add a periodic component to the spot volatility. Taking into account both components, we then proceed with full estimation of the process, and, last, we conclude by testing such an estimation.

\subsection*{Adding a jump component}

Figure \ref{fig:prices_ret} displays the trajectories of the log-average price and the returns. Sudden changes in the price level can be observed, which result in extremely large or extremely small returns, compared to the rest of them. To reflect this behavior, a customary practice is to model log-prices $Y= (Y_t)_{t \geq 0}$ generalizing the semimartingale \eqref{SVmod} by adding a finite activity jump process  $J = (J_t)_{t \geq 0}$, where $J_t = \sum_{j=1}^{N_t} C_j$. The jump process $J$ is assumed to be independent of $B$ and $\tau$. The process $N = (N_t)_{t \geq 0}$ counts the number of jumps that have occurred in the interval $[0, t ]$; and $C= (C_t)_{t \geq 0}$ is a process such that for all $t$, (i) $C_t<\infty$, and (ii) $\sum_{j=1}^{N_t} C_j^2 < \infty$. These properties ensure the quadratic variation of $Y$ is finite. 
For reviews on this generalization see, for instance, \citet{andersen2007roughing}. 

The change in the model affects the measurement of the integrated volatility process, as the realized variance approximates the quadratic variation of $Y$; given by $\tau^*_t+\sum_{j=1}^{N_t} C_j^2$. To deal with this, we first implement a jump-detection procedure, and then perform the measurement of the integrated volatility without taking the jumps into consideration. When deleting the jumps from the sample, the remaining data can be modeled with the GIG-Harris SV model, where the quadratic variation matches the integrated SF-Harris process. Since realized variance approximates quadratic variation also when the observations are not equally spaced, there is no need to replace the jump entries, the crucial part is to detect them.  

Many options to detect jumps are available, we explored a few of them, finding the best results were obtained when using \citet{barndorff2004power} bipower variation, with which we detected around 90 percent of the jumps. The details of the performed procedure can be found in Appendix \ref{App:jump_detection}. After deleting the jumps, we proceeded with the measurement of the integrated and spot volatility processes, $H^*$ and $H$, as described in Section \ref{subsection:SV_model-estimation}. In doing so, the realized variance was computed based on fifteen minutes returns. 

\subsection*{Adding a periodic component}

\begin{figure}[p]
    \centering
	\includegraphics[scale = 0.9]{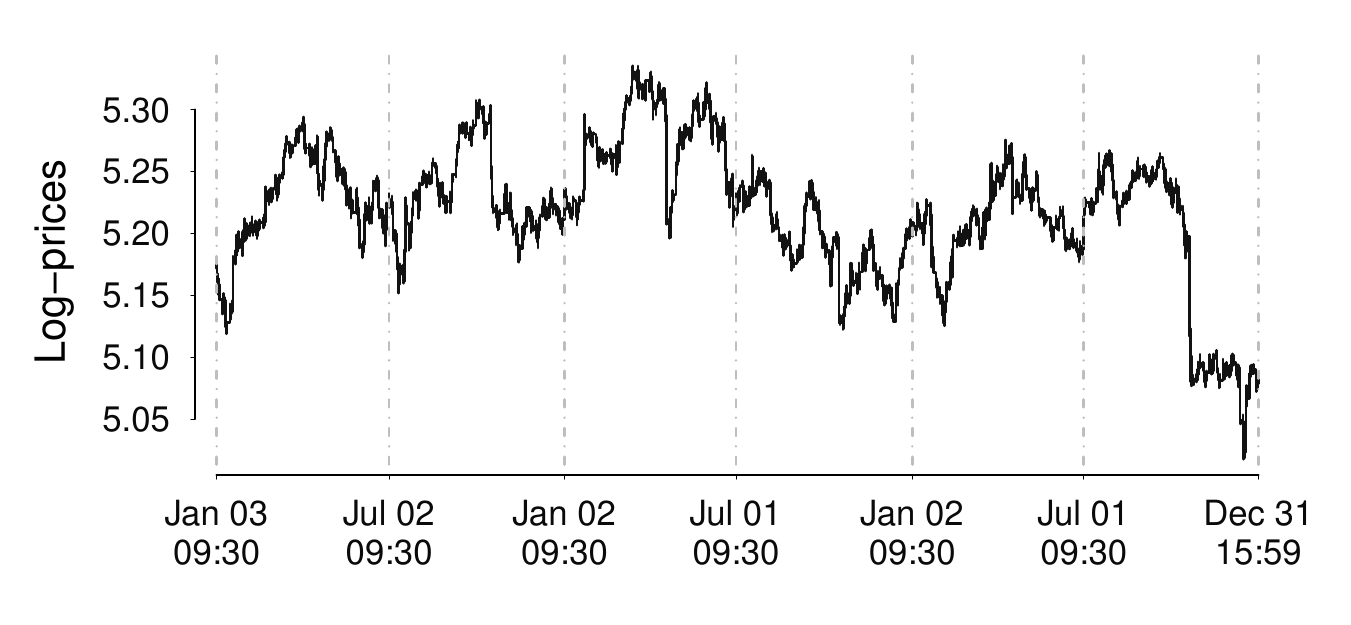}  
	\includegraphics[scale = 0.9]{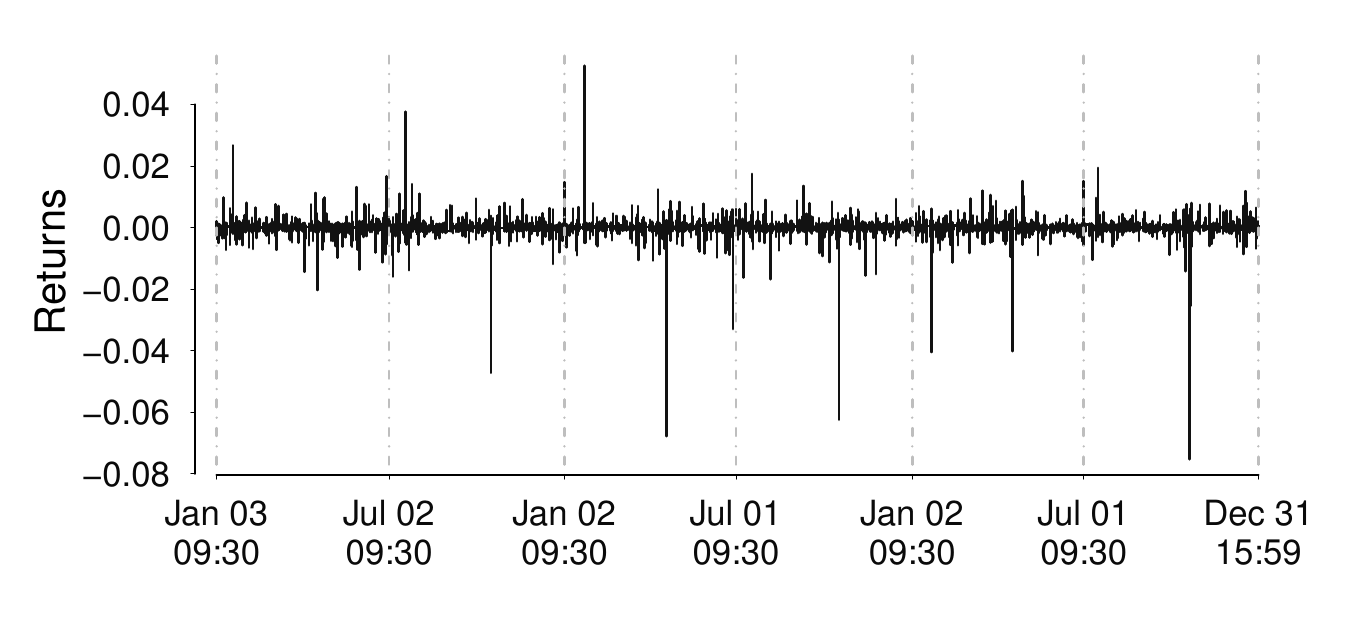}  
	\caption{Log-prices and returns for IBM from 2012 to 2014.}
    \label{fig:prices_ret}
\end{figure} 

\begin{figure}[p]
    \centering
	\includegraphics[scale = 0.9]{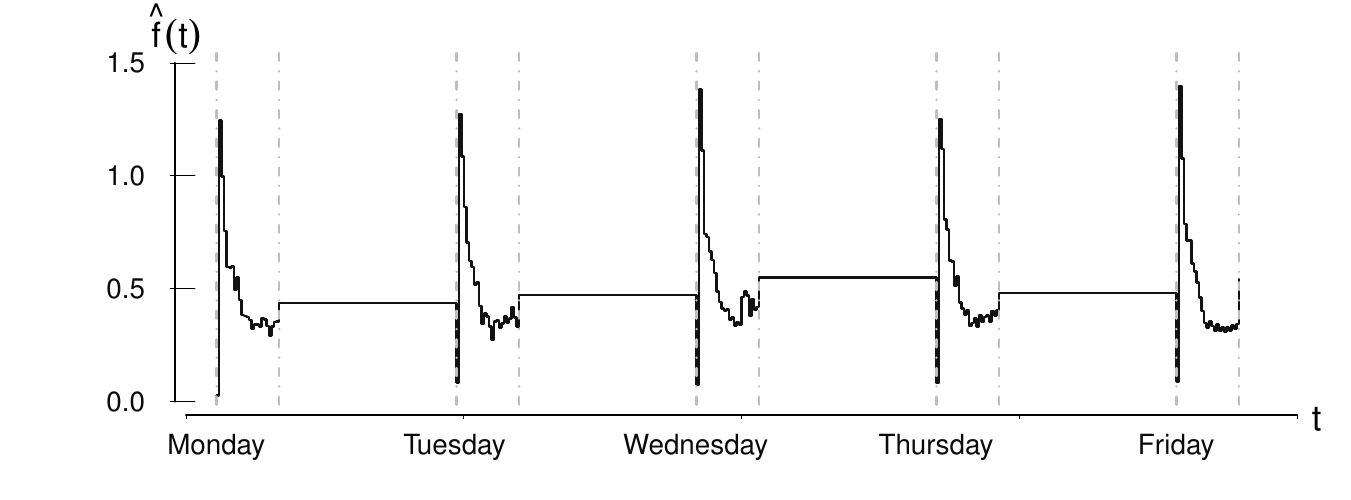}  
	\caption{Approximation of the periodicity function $\hat{f}(t)$ of IBM stock prices.}
    \label{fig:f_week}
\end{figure} 

Recurring events, such as opening, lunch, and closing of financial markets, cause the return volatility to vary systematically over the trading days and weeks. Taking into account this periodic structure may improve the volatility modeling. We present a general procedure to extract the periodic component of a spot volatility process $\tau$; and prove it works for the special case of the SF-Harris process. 

The main idea is based on \citet{boudt2011robust}. Following them, we make a partition of the data time interval in smaller time intervals of length $d$, called local windows; and consider the time transformation $c(s)$ indicating the position of $s$ in the periodicity cycle. So, accordingly, $c(s) = s\bmod L$, where the cycle repeats itself every $L$ days.  Next, we define the periodicity function $f:[0,d] \rightarrow \mathbb{R}^+$ as
\begin{align}\label{Eq:f}
	f(t) = \expect*{\frac{\tau_t}{\frac{1}{d}\int_0^d \tau_s ds}},
\end{align}
and the periodicity factor for each time $t$ as $f\{c(t)\}$. Finally, the process $\hat{\tau}_t = \dfrac{\tau_t}{f\{c(t)\}}$ is called the \textit{periodically adjusted volatility}.  

It is worth noting two important properties. First, when $\tau$ is a stationary process, it suffices to define the function $f$ in the first local window, $[0, d]$, and then extend it to the rest of local windows. Second,  by definition, $\frac{1}{d}\int_0^d f(s)ds = 1. $

To clarify the meaning of $f$ let us think of local windows as days. Then $f$ is the expected value of the spot volatility $\tau$, divided by the average daily volatility $\frac{1}{d}\int_0^d \tau_s ds$. That is, when removing the daily volatility, we would expect that the changes in the remainder volatility were due to the periodicity. This is evident considering $\tau_t$ as a periodic factor depending on $t$ multiplied by a constant effect of the volatility within each day; meaning that the spot volatility, after filtering out the periodicity, is approximately constant over each day. That is the main premise in the studies of \citet{andersen1997intraday}; and \citet{boudt2011robust}; and, with such a premise, is direct to verify that the equality \eqref{Eq:f} is fulfilled.

In the case of the GIG-Harris SV model, the spot volatility after filtering out the periodicity is the SF-Harris process. This implies that, instead of assuming the periodically adjusted volatility constant within each day, we are thinking  of it as a piecewise constant process. Because of Theorem \ref{prop:explicit_form}, taking $\epsilon = 0$, the following proposition ensures equality \eqref{Eq:f} holds in this case. 

\begin{proposition}\label{prop:periodic_comp}
Let $(Y_n)_{n \in \mathbb{N}}$ be a sequence of positive, finite mean, independent identically distributed random variables, stochastically independent of a Poisson process $N = (N_t)_{t \geq 0}$. Let  $d>0$, and $f:[0,d] \rightarrow \mathbb{R}^+$ be a measurable function such that $\frac{1}{d}\int_0^d f(s)ds = 1$. Then, for $t \in [0,d]$, it follows that
\begin{align*}
	\expect*{\frac{f(t)Y_{N_t}}{\frac{1}{d}\int_0^d f(s)Y_{N_s} ds}} = f(t).
\end{align*}
\end{proposition}

The estimation for the periodic component is based on a modification of the sample mean, and is fully developed in Appendix \ref{App:periodicity_estimation}. We set $d$ to one day and $L$ to five, meaning the cycle repeats itself every week of five trading days, as is common in practice. The resulting periodic component graphic appears in Figure \ref{fig:f_week}. 

After this, we continued with the measurement of the periodically adjusted volatility, to which we adjusted the SF-Harris process.  In view of such a process having piece-wise constant paths; when two observations differed by less than a small value $\epsilon$ (we fixed $\epsilon \approx 10^{-5}$), we considered this as measurement noise, and assumed them to be equal to their average.

\subsection{Results}\label{sec:Results}

To test the modeling procedure we broke the sample into an estimation period (first 80 percent of the data), and a subsequent forecasting period (last 20 percent of the data). Next, we predicted probability intervals for the forecasting period with 1000 simulated trajectories. Repeating this for different values of interval probabilities, we computed Table \ref{table:intervals}. If the procedure works correctly, the percentage of the original trajectory that falls within the prediction intervals should be similar to the interval probability. Therefore, the empirical results suggest the GIG-Harris SV model, including a jump and a periodic component, is useful for forecasting. 

\begin{table} 
\caption{\label{table:intervals} Percentage of the trajectory of IBM stock prices that falls within highest posterior density intervals of probability $p$.} 
\begin{center}
	\begin{tabular}{c | c  c  c  c  c  c  c}  
             $p$ & 0.25 & 0.50 & 0.75 & 0.85 & 0.90 & 0.95  \\
             	\hline\\[-3.5ex]
             \% & 25 & 51 & 75 & 84 & 89 & 93 \\
          \hline     
\end{tabular}   
\end{center}
\end{table} 

\section{Concluding remarks and future research directions}\label{section:concluding}

We showed the SF-Harris process is tractable, developing an estimation method and testing it trough rigorous tests. Furthermore, we proved the process provides a high degree of modeling flexibility, coming from the users control over the choice of the processes stationary measure. Next, we proposed a stochastic volatility model based on such a process.

The SV model was developed for the GIG family,  a broad family of marginal distributions. The GIG-Harris SV model shares attractive properties with several popular SV models, it generalizes a model often applied when considering periodicity, yet it has a simple transition mechanism driving its dependence. This allows to understand the model, estimate it, and apply it easily. We applied it to IBM data, proving its predictive strength. 

There are a number of available generalizations for the GIG-Harris SV model to explore. Along the work we mentioned considering the microstructure noise, and exploring jump robust estimation methods for the periodic component. In addition, there is empirical evidence suggesting that, in certain data sets, the dependence on the volatility structure decays at a hyperbolic rate for shorter lags; which is much slower than the exponential time-decay \citep[see, e.g.,][]{andersen1997heterogeneous}. Therefore, an important condition of a short-memory model is that it admits a possible extension to long-memory. 

An extension that has been widely applied is to alter the volatility process by using superpositions. Generalizations of this kind are developed in \citet{barndorff2003integrated}, where the BNS model is extended employing a weighted sum of independent OU-type processes. A first case is to consider a finite number of independent processes.  A new process $\sigma^2$ is defined as a superposition of OU-type processes, $\sigma_1^2,...,\sigma_m^2$, with different persistence rates $\lambda_1,...,\lambda_m$, and independent BDLP $Z_1,...,Z_m$ through
\begin{align} \label{eq:supOU}
	\sigma^2(t) = \sum_{j=1}^m w_j \sigma_j^2(t),
\end{align}
where the positive weights $w_1,...,w_m$ add up to one. The autocorrelation function
\begin{align}\label{eq:corr_sum_exp}
	r(t) = \sum_{j=1}^m w_j e^{-\lambda_j t}
\end{align}
is a weighted sum of exponentials. Thus, some of the volatility components may represent short term variations, while others represent long term movements. Many results from the case $m=1$ carry over to the superposition process. In particular, the integrated process can be derived explicitly. 

In general, it is difficult to find a suitable value of $m$ in equation \eqref{eq:supOU}. An alternative, given by \cite{griffin2010inference}, is to assume an infinite number of components of which only a finite number have non-negligible weight. Now, while this aggregation mechanism provides a possible explanation of the long-range dependence in a time series, the models are still of short-memory. Formally, we shall say that a stationary stochastic process exhibits long-memory if its autocorrelation function has an asymptotic power-like behavior.

A second option to extend the BNS model that may lead to long-memory is to use an infinite number of independent OU-type process. A new process $\sigma^2$ is defined as a superposition of OU-type processes, $\sigma_\eta^2$, with different persistence rates $\lambda_\eta$, and independent BDLP $Z_\eta$ through
\begin{align*}
	\sigma^2(t) = \int_{\mathcal{R}}  \sigma_\eta^2(t) F(d\eta)
\end{align*}
where $F$ is a probability distribution over $\mathcal{R} \subset \mathbb{R^{+}}.$ For example, when $F$ is a $\mathsf{Ga}\{2(1-H), 1\}$ law, the autocorrelation function is given by
\begin{align}\label{eq:corr_gamma}
	r(t) = \left(1+ t \right)^{-2(1-H)}
\end{align}
with $H \in (\frac{1}{2},1)$ being the long memory parameter. 

Aiming for similar extension results with the SF-Harris process, we define the \textit{mixture SF-Harris process} $\chi = (\chi_t)_{t \geq 0}$, as a process given by \eqref{eq:semi-markov}, where $Y$ is as in Theorem \ref{prop:explicit_form}, $(R_n)_{n \in \mathbb{N}}$ is a sequence of random variables independent of $Y$, and the increments $(V_n)_{n \in \mathbb{N}}$, $V_n = R_n-R_{n-1}$, satisfy $V_n \sim \mathsf{Exp}(\rho_n)$, with  $\{\rho_n\}_{n \in \mathbb{N}}$ independent identically distributed random variables with some distribution $F$.

It is direct to verify that 
\begin{align*}
	\mathds{P}_x(\chi_t \in A) = \int_{0}^\infty \{(1-e^{-\rho_1 t})Q(A) + e^{-\rho_1 t} \delta_x(A) \} F(d {\rho_1}) .
\end{align*}
Hence, the transition functions of the mixture SF-Harris process are a mixture of the SF-Harris process transitions. Now, the autocorrelation function of the  mixture SF-Harris process matches the generating function of $F$ since
\begin{align*}
	r(t) = \mathds{P}(V_1>t) &= \int_{0}^\infty e^{-t\rho_1} F(d\rho_1) .
\end{align*}
Therefore,  subject to the dependence structure in a time series, an adequate function $F$ can be chosen. For instance, when $F$ is a degenerated distribution on the value $\alpha$, the mixture SF-Harris process reduces to the SF-Harris process case, where $r(t) = e^{-\alpha t}$. When $F$ is a discrete distribution that takes the value $\lambda_j$ with probability $w_j$ for $j=1,...,m$, we recover the autocorrelation \eqref{eq:corr_sum_exp} of the finite superposition of OU-type processes. Even more, if $F$ is a discrete distribution that may take a countably infinite number of values we get the extension of \cite{griffin2010inference}. Last, when $F$ is a  $\mathsf{Ga}\{2(1-H), t+1\}$, then the autocorrelation function is given by \eqref{eq:corr_gamma}, so it matches the autocorrelation of the infinite superposition of OU-type processes.

More general long-memory models can be obtained by using any heavy-tailed distribution $F$, and, as we saw in Section \ref{section:semi}, a number of stability properties fulfilled by the SF-Harris process, such as wide-sense regeneration, ergodicity, and positive Harris recurrence, follow directly for the general model.  Also, the integrated process can be derived explicitly. However, the process is neither uniformly nor exponentially ergodic. This makes sense since observations far away in the past remain correlated.

Lastly, notice that this study dealt with the one-dimensional case. However, the SV model definition, and its respective research, are extended to multidimensional cases in a natural way. 

\if 0\blind
{
\section*{Acknowledgments}
We gratefully acknowledge the National Council for Science and Technology of Mexico (CONACyT). The first author was supported by a national CONACyT scholarship, while the second author was supported by the CONACyT grant number 241195.
} \fi

\bigskip
\appendix

 \section{SF-Harris process}\label{App:Harris_process}
 
 \subsection{Stability properties of the semi-Markovian extension}\label{App:stability}
 
We mentioned the process $\xi = (\xi_t)_{t \geq 0}$ given by \eqref{eq:semi-markov} has a set of intervals of constancy in their trajectory and, since this structure is similar to the SF-Harris process, we can deduce a number of properties. In particular, that $Q$ is the limit distribution. Furthermore, 
	\begin{align*}
	\sup_{A \in \mathcal{E}} |\mathds{P}_x(\xi_t \in A)-Q(A)| 
	&= \{1-G(t)\}  \sup_{A \in \mathcal{E}} |\delta_x(A)-Q(A)| \leq 1-G(t).
\end{align*}
Hence, 
	\begin{equation*}
		\lim_{t \rightarrow \infty} \left( \sup_{A \in \mathcal{E}} |\mathds{P}_x(\xi_t \in A)- Q(A)|\right)  = 0
	\end{equation*}
and the process is always ergodic.  If $G$ is a light-tailed distribution, meaning its tail is exponentially bounded, then the process is uniformly ergodic.

	 When $Q(B)>0$, it is easy to verify that
 \begin{equation*}
	 \mathds{P}_x\left\{ \int_0 ^\infty \delta_{\xi_t}(B)dt = \infty\right\} \geq \mathds{P}_x\left\{  \sum\limits_{n=1}^\infty V_n \delta_{Y_n}(B) = \infty\right\} = 1.
\end{equation*}
Therefore, the process  \eqref{eq:semi-markov} is positive Harris recurrent. Additionally, we have the following results, whose proofs are analogous to the corresponding SF-Harris process case.

\begin{theorem}\label{rep_semi-Markov}
Let $\xi = (\xi_t)_{t \geq 0}$ be as in \eqref{eq:semi-markov}.  If $K_t = \max \{n \in \mathbb{N}: R_n \leq t\}$,  and $K_0 = 0$, then $K =  (K_t)_{t \geq 0}$ is stochastically independent of $Y$ and $\xi_t = Y_{K_t}$ for all $t \geq 0$.
\end{theorem}
 Notice that $K$ is a renewal counting process representing the number of arrivals of the process  \eqref{eq:semi-markov} in the interval $(0,t]$.
 
 \begin{theorem} \label{rep_semi-Markov_int}
	Let $\xi = (\xi_t)_{t \geq 0}$ be as in \eqref{eq:semi-markov}; and let K be as in Theorem \ref{rep_semi-Markov}. The integrated process, $\xi^* = (\xi_t^*)_{t \geq 0}$, where $\xi_t^* = \int_0^t \xi_s ds$,  is given by
\begin{align*}
	\xi_t^* = \sum\limits_{n=0}^{K_t-1} (Y_n-Y_{K_t}) V_{n+1} + Y_{K_t}t,
\end{align*}
taking $Y_0 = x$.
\end{theorem}

\section{Estimation}\label{App:estimation}

\subsection{Development of estimation methods for the SF-Harris process when Q is a GIG distribution} \label{App:estGIG} 

\vspace*{.5em}
\subsubsection*{The no difference no jump method (NDNJ)}
We locate the set $J=\{j:x_j \neq x_{j-1}\}$. If $J$ is empty, then $\hat{\alpha}=0$, otherwise $\hat{\alpha}=\left(\frac{1}{|J|}\sum_{j \in J} t_j\right)^{-1}$. To obtain $(\hat{\lambda}, \hat{\kappa},\hat{\eta})$ we maximize numerically 
\begin{align*}
	L(\lambda, \kappa, \eta) = \sum_{j \in J} \log \mathsf{GIG}(x_j|\lambda, \kappa, \eta)
\end{align*}
using the BFGS quasi-Newton algorithm (named after Broyden, Fletcher, Goldfarb, and Shanno). Such an algorithm has proven to have good performance for non-smooth optimizations \citep[see, e.g.,][]{nocedal2006numerical}.

\subsubsection*{Maximum likelihood estimation (MLE)}

The corresponding log-likelihood is four-dimensional. Because of the convoluted form of its gradient we apply the Nelder-Mead method \citep{Nelder1965simplex} to maximize it.

\subsubsection*{Expectation-maximization algorithm (EM)}
	
It remains to expand the expression $\sum \limits_{i=0}^n p_i^{(k-1)}\log Q(x_i|\beta)$, which in this case is proportional to
\begin{align*}
	\sum \limits_{i=0}^n p_i^{(k-1)} \left[ - \log\{K_{\lambda}(\kappa)\} - \lambda \log(\eta) + (\lambda-1) \log(x_i)  -\frac{\kappa}{2}\left(\frac{\eta}{x_i} + \frac{x_i}{\eta}\right)\right]
\end{align*}
or, alternatively,
\begin{align*}
	- m^{(k-1)}\log\{K_{\lambda}(\kappa)\} - m^{(k-1)} \lambda \log(\eta)+ (\lambda-1) S_1^{(k-1)}  -\frac{\kappa}{2}\left(\eta S_2^{(k-1)} + \frac{1}{\eta} S_3^{(k-1)}\right),
\end{align*}
where $m^{(k-1)} = \sum\limits_{i = 1}^n p_i^{(k-1)}$, $S_1^{(k-1)} = \sum\limits_{i = 1}^n p_i^{(k-1)}\log(x_i)$, $S_2^{(k-1)} = \sum\limits_{i = 1}^n p_i^{(k-1)}/x_i$, and $S_3^{(k-1)} = \sum\limits_{i = 1}^n p_i^{(k-1)}x_i$.\\

We maximize numerically at each iteration, using BFGS algorithm.

\subsubsection*{Gibbs sampler}

We assume $\pi(\lambda, \kappa, \eta) = \pi(\lambda)\pi(\kappa)\pi(\eta)$, with $\lambda \sim \mathsf{N}(\mu_\lambda, \sigma^2_\lambda)$, $\kappa \sim \mathsf{Ga}(a_\kappa, b_\kappa)$, and $\eta \sim \mathsf{Ga}(a_\eta, b_\eta)$. The final joint log-density function is
\begin{align*}
	\log \pi(\lambda, \kappa, \eta | \alpha, \textbf{x}, \textbf{z}) &=  \sum_{i=1}^n z_i \log Q(x_i|\beta) + \log \pi(\lambda)+ \log \pi(\kappa)+ \log \pi(\eta).
\end{align*}

With an analogous calculation to the one done in the EM method we obtain
\begin{align*}
	\sum \limits_{i=0}^n z_i\log Q(x_i|\beta) &= - m\log\{K_{\lambda}(\kappa)\} - m \lambda \log(\eta)+ (\lambda-1) S_1 -\frac{\kappa}{2}\left(\eta S_2 + \frac{1}{\eta} S_3\right),
\end{align*}
where $m=\sum_{i = 1}^n z_i$, $S_1 = \sum_{i = 1}^n z_i\log(x_i)$, $S_2 = \sum_{i = 1}^n z_i/x_i$, and $S_3 = \sum_{i = 1}^n z_ix_i$. Therefore, the full conditional log-distributions are:
\begin{align*}
	\log \pi(\lambda | \alpha, \kappa, \eta, \textbf{x}, \textbf{z}) &=  - n\log\{K_{\lambda}(\kappa)\} - \lambda n \log(\eta) + (\lambda-1) S_1  - \frac{1}{2\sigma^2_\lambda}(\lambda-\mu_\lambda)^2,\\
	\log \pi(\kappa | \alpha, \lambda, \eta, \textbf{x}, \textbf{z}) &= - n\log\{K_{\lambda}(\kappa)\}  -\frac{\kappa}{2}\left(\eta S_2 + \frac{1}{\eta} S_3\right)  + (a_\kappa-1)\log(\kappa) - b_\kappa \kappa,\\
	\log \pi(\eta | \alpha, \lambda, \kappa, \textbf{x}, \textbf{z}) &= - \lambda n \log(\eta)  -\frac{\kappa}{2}\left(\eta S_2 + \frac{1}{\eta} S_3\right)  + (a_\eta-1)\log(\eta) - b_\eta \eta.
\end{align*}
We use the Adaptive Rejection Metropolis Sampling method to simulate them.

\subsection{KL divergence of GIG distributions} \label{App:KL} 

\begin{figure}[p]
    \centering
	\includegraphics[scale = 0.75]{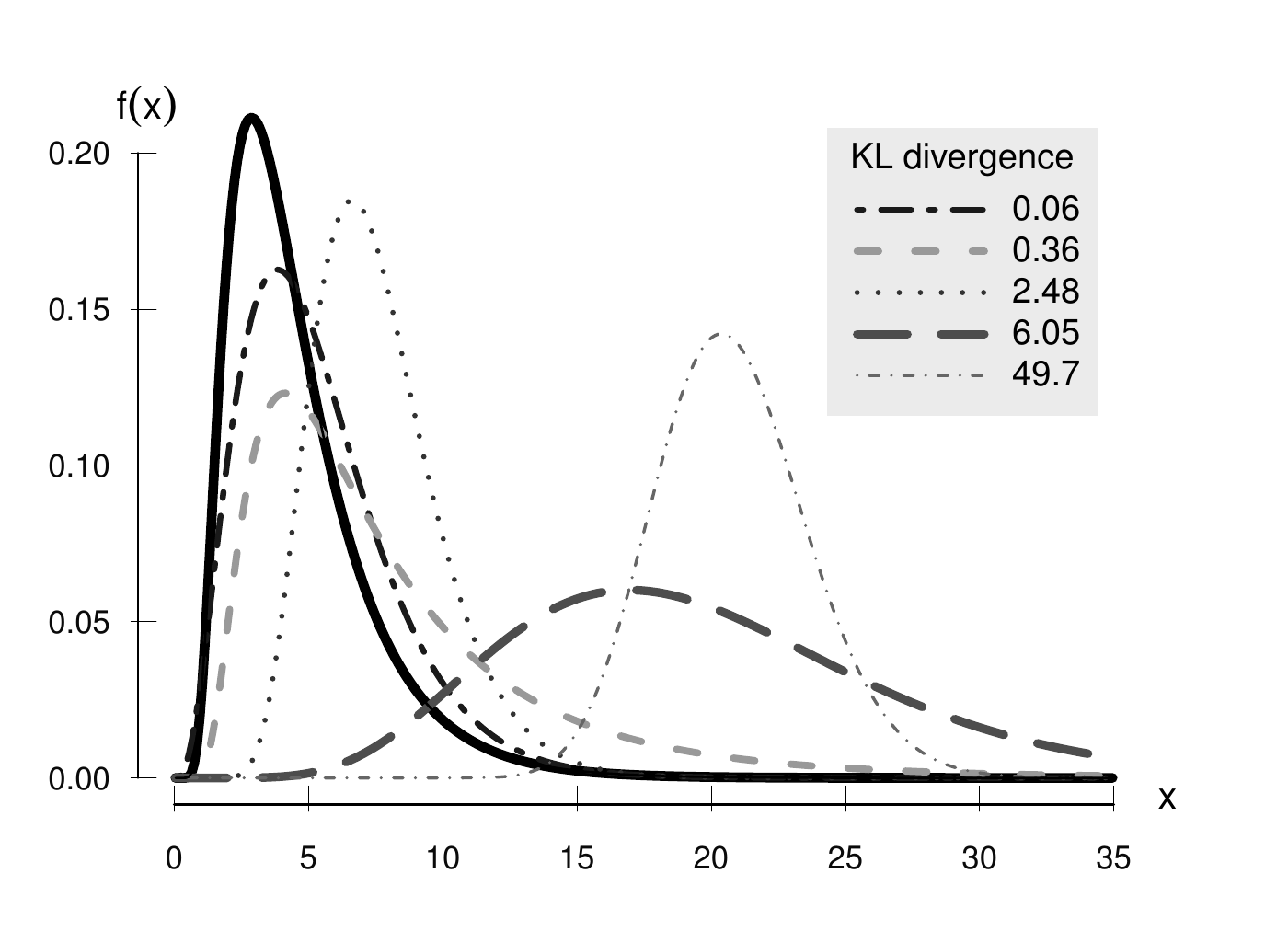}   
	\caption{In solid appears the $\mathsf{GIG}(0,3,4)$ density function. The dotted lines are different approximations of the GIG class to such a distribution along with their $\mathsf{KL}$ divergence.}
    \label{fig:KL}
\end{figure} 

The Kullback-Leibler divergence $\mathsf{KL}(p, \hat{p})$ is a non-symmetric similarity measure between two probability distribution functions $p$ and $\hat{p}$ defined as
\begin{align*}
	\mathsf{KL}(p, \hat{p})= \int_\mathcal{X} p(x) \log \frac{p(x)}{\hat{p}(x)} dx.
\end{align*}

We use it as a measure of the amount of information lost by working with the estimated model instead of the real one. If $p(x) = \mathsf{GIG}(x|\lambda, \kappa, \eta)$ and $\hat{p}(x) = \mathsf{GIG}(x|\hat{\lambda}, \hat{\kappa}, \hat{\eta})$, then
\begin{align*}
	\frac{p(x)}{\hat{p}(x)} &= \frac{K_{\hat{\lambda}}(\hat{\kappa})}{K_{\lambda}(\kappa)} \frac{\hat{\eta}^ {\hat{\lambda}}}{\eta^\lambda} x^{\lambda-\hat{\lambda}}  \exp \left[ -\frac{1}{2} \left\{(\kappa \eta -\hat{\kappa} \hat{\eta})\frac{1}{x} + \left( \frac{\kappa }{\eta} - \frac{\hat{\kappa} }{\hat{\eta}} \right) x\right\} \right]. 
\end{align*}
It follows that
\begin{align*}
	\mathsf{KL}(p, \hat{p}) = f(\lambda, \kappa,& \eta, \hat{\lambda}, \hat{\kappa}, \hat{\eta}) + (\lambda-\hat{\lambda}) \expect*{\log X} \\ &-\frac{1}{2} \left\{(\kappa \eta -\hat{\kappa} \hat{\eta})\expect*{\dfrac{1}{X}} + \left( \frac{\kappa }{\eta} - \frac{\hat{\kappa} }{\hat{\eta}} \right) \expect*{X}\right\},
\end{align*}
where $X \sim p$, and $f(\lambda, \kappa, \eta, \hat{\lambda}, \hat{\kappa}, \hat{\eta}) = \log\{K_{\hat{\lambda}}(\hat{\kappa})\} - \log \{K_{\lambda}(\kappa)\} + \hat{\lambda} \log (\hat{\eta}) - \lambda \log (\eta)$.

Making $R_{\lambda}(\kappa) = K_{\lambda+1}(\kappa)/K_{\lambda}(\kappa)$, it is known that: 
\begin{itemize}
	\item[(i)] $\expect*{X} = R_{\lambda}(\kappa)\eta$.
	\item[(ii)] $\expect*{X^{-1}} = \eta^{-1}\left( R_{\lambda}(\kappa)-2\lambda \kappa ^{-1}\right)$.
	\item[(iii)] $\expect*{\log(X)} = \log(\eta) + \dfrac{d}{d\lambda} K_{\lambda}(\kappa)\left\{K_{\lambda}(\kappa)\right\}^{-1}.$
\end{itemize}

\begin{table}
	\caption{\label{table:KL}KL divergence between $p(x) = \mathsf{GIG}(x|0, 3, 4)$ and $\hat{p}(x) = \mathsf{GIG}(x|\hat{\lambda}, \hat{\kappa}, \hat{\eta})$}  
	\begin{center}
	\begin{tabular}{c | c c c c c}  
		$\mathsf{KL}(p, \hat{p})$ & 0.06 &  0.36 &  2.48 & 6.05 & 49.75\\ [.3ex] 
		\hline\\[-3.5ex] 
		$\hat{\lambda}$ & 3.5 & -1 & 2 & 7 & 50 \\
		$\hat{\kappa}$ & 0.8 & 2 & 10 & 3 & 20 \\
		$\hat{\eta}$ & 0.6 & 10  & 6  & 4  & 4 \\
		\hline
	\end{tabular}  
	\end{center}
\end{table}

As an example, we set $\lambda = 0$, $\kappa = 3$, and $\eta =4$, and compute the KL divergence for different approximations (Table \ref{table:KL} and Figure \ref{fig:KL}). We can observe that in certain situations, although the approximated parameters are extremely different from the real one, the distribution approximation is good. Because of this, the divergence is used to measure the fit, rather than an established measure for each parameter.

\subsection{Final density of \texorpdfstring{$\mu$} \ \ and \texorpdfstring{$\beta$}{Lg}} \label{App:estimation_mu_beta}

To facilitate the calculations display we assume that times are equidistant and $h$ denotes the distance between them. We also assume $\pi(\mu,\beta) = \pi(\mu)\pi(\beta)$, with $\mu \sim \mathsf{N}(m_\mu, \sigma_\mu^2)$, and $\beta \sim \mathsf{N}(m_\beta, \sigma_\beta^2)$. Their joint posterior density is given by
\begin{align*}
	\pi(\mu, \beta | R_1,...,R_{n-1}) = f(R_1,...,R_{n-1}|\mu,\beta)\pi(\mu, \beta),
\end{align*}
where
\begin{align*}
	f(R_1,...,R_{n-1}|\mu,\beta) &= f(R_1,...,R_{n-1}|H_1^*,...,H_{n-1}^*, \mu,\beta)f(H_1^*,...,H_{n-1}^*| \mu,\beta)\\
		&\propto  \exp \left\{ \sum_{i=1}^{n-1} -\dfrac{1}{2H_i^*}(R_i-\mu h - \beta H_i^*)^2 \right\},
\end{align*}
since $R_1|H_1^*,...,R_{n-1}|H_{n-1}^*$ are independent, and $H_1^*,...,H_{n-1}^*$ do not depend on $\mu$ or $\beta$.\\

After some algebra we obtain that, making $\bar{R} = \frac{1}{n-1}\sum\limits_{i=1}^{n-1} R_i$, $\bar{R_2} = \frac{1}{n-1}\sum\limits_{i=1}^{n-1} \dfrac{R_i}{H_i^*}$, $\bar{H}^* = \frac{1}{n-1}\sum\limits_{i=1}^{n-1}  H_i^*$, and $\bar{H}^*_2 = \frac{1}{n-1} \sum\limits_{i=1}^{n-1} \dfrac{1}{H_i^*}$, the last expression becomes
\begin{align*}
	f(R_1,...,R_{n-1}|\mu,\beta) &\propto \exp \left\{ -\dfrac{1}{2} \left( \mu^2 h^2 \bar{H}^*_2 -2  \mu h \bar{R_2} +2 h \mu \beta    + \beta^2 \bar{H}^* -2  \beta \bar{R}  \right) \right\}.
\end{align*}
Since
\begin{align*}
	\pi(\mu) \propto \exp \left\{-\dfrac{1}{2\sigma_\mu^2} \left(\mu^2-2\mu m_\mu\right) \right\}\ \text{and} \
	\pi(\beta) \propto \exp \left\{-\dfrac{1}{2\sigma_\beta^2} \left(\beta^2-2\beta m_\beta\right) \right\},
\end{align*}
taking $A = h^2 \bar{H}^*_2+ (\sigma_\mu^2)^{-1}$, \ $B = h \bar{R_2} + (m_\mu)(\sigma_\mu^2)^{-1}$, \ $C = h$, \ $D = \bar{H}^* + (\sigma_\beta^2)^{-1}$, \ and $E = \bar{R} + (m_\beta)(\sigma_\beta^2)^{-1}$, we get
\begin{align*}
	\pi(\mu, \beta | R_1,...,R_{n-1}) &\propto \exp \left\{ -\dfrac{1}{2} \left( \mu^2 A -2  \mu B +2 \mu \beta  C +  \beta^2 D -2  \beta E \right) \right\}.
\end{align*}
Now marginalizing,
\begin{align*}
	\pi(\mu | R_1,...,R_{n-1})  &\propto \exp\left\{ -\dfrac{1}{2} \left(  \mu^2 A -2  \mu B \right) \right\} \int \exp \left[ -\dfrac{1}{2} \left\{ \beta^2 D - 2\beta ( E -  \mu C) \right\} \right] d\beta\\
	&\propto \exp\left\{ -\dfrac{1}{2} \left(  \mu^2 A -2  \mu B \right) \right\}   \exp \left\{-\dfrac{1}{2} \left(- \frac{\mu^2 C^2 -2\mu E C}{D}  \right) \right\}\\
	&\propto \exp\left[ -\dfrac{1}{2} \left(A - \frac{C^2}{D}\right) \left\{  \mu   - \left(  \frac{DB  - E C}{AD - C^2} \right)   \right\}^2 \right].
\end{align*}
So if we assign $F = AD - C^2$, then $\mu | (R_1,...,R_{n-1}) \sim N(\tilde{m_\mu}, \tilde{\sigma^2_\mu})$, with $\tilde{m_\mu} = (DB  - E C)F^{-1}$, and $\tilde{\sigma^2_\mu} = DF^{-1}$. With an analogous calculation it can be proven $\beta | (R_1,...,R_{n-1}) \sim N(\tilde{m_\beta}, \tilde{\sigma^2_\beta})$, with $\tilde{m_\beta} = (EA  - BC)F^{-1}$, and $\tilde{\sigma^2_\beta} = AF^{-1}$.

\section{Empirical Analysis}\label{App:empirical_analysis}

\subsection{Cleaning procedure}\label{App:cleaning_procedure}

To clean the data we implemented the step-by-step procedure proposed by \citet{barndorff2009realized}.  Part of the procedure is to delete (or replace) entries with a repeated time stamp, a zero transaction price, or a time stamp outside the exchange hours. No entries were found in any of these cases. Hence, we proceeded by creating a single variable that represents the stock price. We computed the standard deviation between the open and close price for each observation, finding is less than 0.1 for 95 percent of the sample. This suggests that creating such a new variable by averaging the open and close prices makes sense, we call it \textit{average price}.

Two remaining steps of the cleaning procedure were run using the average price. The first one is to delete entries for which the average price exceeds by more than 50 times the median on that day. Once again no entries were removed. The second step is to replace entries for which the average price deviated by more than 10 mean absolute deviations from a rolling centered median (excluding the observation under consideration) of 50 observations (25 observations before and 25 after). This was performed in a slightly different way, restricting the observations from the rolling median to be on the same day. A total of 154 entries was replaced in this step, which is around 0.05 percent of the sample size.  

\subsection{Jump detection}\label{App:jump_detection}

To detect the jumps we use the bipower variation. The bipower variation, introduced in \citet{barndorff2004power}, equals the quadratic variation of the continuous component and, in a range of cases, it produces an estimator of integrated volatility in the presence of jumps. It can be consistently estimated, and the estimator is called realized bipower variance.  

The difference between the realized variance and the realized bipower variance is, consequently, a good indicator for the squared jumps; but nothing prevents this difference from becoming negative in a given finite sample. Thus, following \citet{barndorff2004power} suggestion, we calculate at every fifteen minutes the maximum value between such a difference and zero. The top 0.1 percent values are considered intervals with jumps. We explore individually those fifteen minute intervals, marking the entry that differs the most from the interval mean as a jump. Notice we mark exactly one jump in each interval. This could be restrictive and is fixed by running the procedure a few times. We found that in these data it suffices to run it twice, so around 0.2 percent of the sample observations were considered as jumps and deleted. 

\subsection{Periodic component estimation}\label{App:periodicity_estimation}

The most natural method for the periodic component estimation is to approximate the expected value with the sample mean, 
\begin{align*}\label{eq:f_estimator}
	\hat{f}(t) = \dfrac{1}{|G(t)|} \sum_{r \in G(t)} \dfrac{\tau_r}{\frac{1}{d}\int_{w(r)} \tau_s ds},
\end{align*}
with $G(t) = \{r: c(r)=t \}$, and $w(r)$ the local window which contains $r$; standardizing such an approximation so it meets the condition $\frac{1}{d}\int_0^d f(s)ds = 1$. Nonetheless, this estimator may be biased in the presence of jumps. Since we have eliminated the sample jumps, we can apply it without such concern. Still if, for a given $t$, the standard deviation of the $G(t)$ values exceeded the standard deviation of the sample; we computed the mean with those values lying below the $0.9$-quantile.  It is noteworthy that, when evidence of jump presence is found, robust estimation methods can be further explored; extending the ideas exposed in \citet{boudt2011robust} to the GIG-Harris process.

\section{Proofs}\label{App:proofs}

\begin{proof}[Theorem \ref{repH_int}] 
\begin{align*}
	H^*_t &=  \int_0^t \{ x \mathds{1}_{s<T_1}+\sum\limits_{n=1}^\infty Y_n \mathds{1}_{s\in {[T_n, T_{n+1})} } \} \ ds \\
		&=  x \int_0^{t \land T_1} ds + \sum\limits_{n=1}^\infty Y_n \int_{t \land T_n}^{t \land T_{n+1}} ds.
\end{align*}
Since $Y_0 = x$ and $T_0 = 0$, it follows
\begin{align*}
	H^*_t =  \sum\limits_{n=0}^\infty Y_n \{(t \land T_{n+1}) - (t \land T_n)\}. 
\end{align*}
Equivalently, 
\begin{align*}  
	H^*_t &=  \sum\limits_{n=0}^{N_t-1} Y_n (T_{n+1}-T_n) + Y_{N_t}(t-T_{N_t})\\
		&= \sum\limits_{n=0}^{N_t-1} Y_n S_{n+1} + Y_{N_t}\left(t-\sum\limits_{n=0}^{N_t-1} S_{n+1}\right) \\
		&= \sum\limits_{n=0}^{N_t-1} (Y_n-Y_{N_t}) S_{n+1} + Y_{N_t}t.
\end{align*}
\end{proof}

\begin{lemma} \label{lemma1} 
For any sequence of independent identically distributed random variables $X_1,...,X_n$, and any $a_1,...,a_n \in \mathbb{R}$ whose sum is nonzero, it holds that 
\begin{align*}
\expect*{\dfrac{X_j}{\sum_{i=1}^{n}a_iX_i}} = \dfrac{1}{\sum_{i=1}^n a_i}, \  \ j=1,...,n.
\end{align*}
\end{lemma}

\begin{proof}[Lemma \ref{lemma1}] 
Let 
\begin{align*}
	k= \expect*{\dfrac{X_j}{\sum_{i=1}^{n}a_iX_i} }
\end{align*}
 for all $j=1,...,n$, then,
 \begin{align*}
	 1 = \expect*{\frac{\sum_{i=1}^n a_iX_i}{\sum_{i=1}^{n}a_iX_i} } = \sum_{i=1}^{n} a_i\expect*{\frac{X_i}{\sum_{i=1}^{n}a_iX_i} } = k\sum_{i=1}^{n} a_i.
\end{align*}
\end{proof}

\begin{proof}[Proposition \ref{prop:periodic_comp}] 
The case $N_d = 0$ is trivial, hence, we assume $N_d>0$. First, we consider that $f$ is a measurable simple function, meaning that 
\begin{align*}
	f(t)=\sum_{j=1}^{m}B_jI_{(b_{j-1}, b_j]}(t),
\end{align*}
where $0=b_0<b_1<\cdots<b_{m-1}<b_m=d$, and $B_1,...,B_m  \in \mathbb{R}$. Then, it follows that, taking $\mathcal{F}_d = \sigma(	N_s:s \leq d)$, 
\begin{align*}
	\expect*{\frac{Y_{N_t}}{\frac{1}{d}\int_0^d f(s)Y_{N_s} ds}} = \expect*{ \expect*{\frac{Y_{N_t}}{\frac{1}{d} \int_0^d f(s)Y_{N_s} ds} | \mathcal{F}_d } }.
\end{align*}
Now, denoting with $S_1,S_2,...$ the time between jumps of $N$, given $\mathcal{F}_d$, let $n = N_d$, We define $J_0 = 0$, $J_i = \sum_{k=0}^{i-1} S_k$ for  $i = 1,...,n-1$, and $J_n = d$.
Next, we create a new partition of the interval $[0,d]$ by taking the common refinement of partitions $0=b_0<\cdots<b_m=d$ and $0=J_0<\cdots<J_n=d$ (consisting of all different points from the two partitions renamed in order). Suppose we get the partition $0=v_0<\cdots<v_r=d$. There is a representation of $f$ in terms of such a partition,
\begin{align*}
	f(t) = \sum_{j=1}^{r}\hat{B}_jI_{(v_{j-1}, v_j]}(t),
\end{align*}
and, consequently,
\begin{align*}
	\int_0^d f(s)Y_{N_s}ds = \sum_{j=1}^r \hat{B}_jY_j(v_j - v_{j-1}).
\end{align*}

Moreover, $N_t \in \{0,...,N_d\}$, so $S_0,...,S_{N_d-1}$, $N_t$, and $N_d$ are $\mathcal{F}_d$-measurable. Therefore, applying Lemma \ref{lemma1} we get
\begin{align*}
	 \expect*{\frac{Y_{N_t}}{\frac{1}{d} \int_0^d f(s)Y_{N_s} ds} | \mathcal{F}_d  } = \frac{1}{\frac{1}{d} \sum_{j=1}^{r}\hat{B}_j (v_j-v_{j-1})},
\end{align*}
but
\begin{align*}
	\frac{1}{d} \sum_{j=1}^{r}\hat{B}_j (v_j-v_{j-1}) = \frac{1}{d} \int_0^d f(s)ds = 1.
\end{align*}
So we obtain the desired equality for measurable simple functions. 

Now, if $f$ is any measurable function, then there is an increasing sequence $(f_n)_{n \in \mathbb{N}}$ of measurable simple functions which converge punctually to $f$ (almost surely). Thus, as is customary in proofs, the desired equality follows using the case proven for measurable simple functions, and applying the Monotone Convergence Theorem.
\end{proof}

\clearpage
\bibliographystyle{chicago}
\begin{small}
	\bibliography{HarrisProcess_Anzarut_Mena}
\end{small}

\end{document}